\documentclass[11pt,a4paper,reqno,oneside]{amsart}
\usepackage[english]{babel}
\usepackage{stix2}
\usepackage{microtype}
\usepackage{ragged2e}
\usepackage[autostyle=true]{csquotes} 

\usepackage{zref-clever}
\numberwithin{equation}{section}
\zcsetup{
    cap=true, 
    nameinlink=false, 
}

\zcRefTypeSetup{equation}{
    abbrev=true, 
}

\newcommand{\cref}[1]{\zcref{#1}}
\newcommand{\Cref}[1]{\zcref[S]{#1}}

\usepackage{graphicx} 
\graphicspath{{./figures/}}

\usepackage[listformat=empty]{caption} 
\usepackage{subcaption}
\usepackage[shortlabels]{enumitem}

\usepackage{keytheorems} 
\newkeytheorem{theorem}[
    parent=section
]
\newkeytheorem{proposition,lemma,corollary}[
    style=plain,
    sibling=theorem 
]
\newkeytheorem{definition}[
    style=definition,
    sibling=theorem
]
\newkeytheorem{remark,example}[
    style=remark,
    sibling=theorem
]
\newkeytheorem{hypothesis}[
    style=plain,
    sibling=theorem,
    name=Hypothesis,
    refname=hypothesis,
    Refname=Hypothesis,
]

\usepackage[
backend=bibtex,
style=numeric-comp,
giveninits=true,
maxbibnames=299,
natbib=true,
doi=false,
isbn=false,
url=false,
sorting=nyt,
maxbibnames=8,
eprint=false
]
{biblatex}
  \renewbibmacro{in:}{}
\DeclareFieldFormat[misc]{title}{``#1''\isdot}
\DeclareFieldFormat[online]{title}{``#1''\isdot}
\DeclareFieldFormat{journaltitle}{#1\isdot}
\DeclareFieldFormat[article,periodical]{volume}{\mkbibbold{#1}}
\DeclareFieldFormat{pages}{#1}
 \AtEveryBibitem{%
 	\clearfield{month}}
    \AtEveryBibitem{%
 	\clearfield{day}}
      \AtEveryBibitem{%
 	 \clearfield{number}}
\AtEveryBibitem{%
 	 \clearfield{pagetotal}}
\renewbibmacro*{addendum+pubstate}{} 

\AtEveryBibitem{%
	\clearfield{number}}    \AtEveryBibitem{\clearfield{pagetotal}}    \AtEveryBibitem{\clearfield{volumes}}
\AtEveryBibitem{\clearfield{language}}

    \AtEveryBibitem{%
        \iffieldundef{doi}
        {
            \iffieldundef{issn}
            {}
            {\clearfield{isbn}}
        }
        {
            \clearfield{issn}%
            \clearfield{isbn}
        }%
    }

\newbibmacro{string+doiurleprint}[1]{%
  \iffieldundef{doi}{%
    \iffieldundef{url}{%
      \iffieldundef{eprint}{#1}%
      {\href{https://arxiv.org/abs/\thefield{eprint}}{#1}}}%
    {\href{\thefield{url}}{#1}}}%
  {\href{https://doi.org/\thefield{doi}}{#1}}}
\DeclareFieldFormat{title}{\usebibmacro{string+doiurleprint}{\mkbibemph{#1}}}
\DeclareFieldFormat[article,periodical,inbook,incollection,inproceedings,thesis,report]{title}%
  {\usebibmacro{string+doiurleprint}{``#1''\isdot}}
\DeclareFieldFormat[misc,online]{title}{\usebibmacro{string+doiurleprint}{``#1''\isdot}}

\usepackage{appendix}

\usepackage{mathtools}
\usepackage{amsfonts}
\usepackage{amssymb}
\usepackage{dsfont}
\usepackage{nicefrac}
\usepackage{braket}
\usepackage{diffcoeff}
\usepackage{siunitx}
\usepackage{cancel}

\usepackage{xcolor}
\definecolor{royalpurple}{rgb}{0.47, 0.32, 0.66}
\definecolor{dblue}  {RGB}{20,66,129}
\definecolor{ddblue} {RGB}{11,36,69}
\definecolor{jblue}  {RGB}{20,50,100}
\definecolor{cblue}{rgb}{0.16, 0.32, 0.75}
\definecolor{cred}{rgb}{0.7, 0.11, 0.11}

\usepackage{hyperref} 
\hypersetup{%
	colorlinks = true,
	linkcolor=cred,
	citecolor=cblue,
	urlcolor=black,
}

\newcommand{\id}{\mathds{1}}

\DeclareMathOperator{\Real}{Re}
\DeclareMathOperator{\Imag}{Im}
\DeclareMathOperator{\mspan}{Span}
\DeclareMathOperator{\Ker}{Ker}
\DeclareMathOperator{\Ran}{Ran}

\DeclarePairedDelimiterX{\norm}[1]\lVert\rVert{
	\ifblank{#1}{\:\cdot\:}{#1}
}
\DeclarePairedDelimiterX{\opnorm}[1]\lVert\rVert{
	\ifblank{#1}{\:\cdot\:}{#1}
}

\newcommand{\iu}{\mathrm{i}\mkern1mu}

\newcommand{\nnum}{\mathbb{N}}
\newcommand{\rnum}{\mathbb{R}}
\newcommand{\cnum}{\mathbb{C}}

\newcommand{\hilbert}{\mathcal{H}}
\newcommand{\domain}{\mathcal{D}}

\newcommand{\bounded}{\mathcal{B}}
\newcommand{\Uop}{\mathcal{U}}

\newcommand{\Jop}{\mathcal{J}}

\newcommand{\Jmin}{J_{\mathrm{min}}}

\newcommand{\Hmin}{H_{\mathrm{min}}}

\newcommand{\variation}{\mathcal{V}}

\newcommand{\hmat}{H_{\mathrm{mat}}}
\newcommand{\aint}{\Sigma}
\newcommand{\pint}{\Pi}
\newcommand{\dspace}{\mathcal{N}}

\numberwithin{equation}{section}

\newkeytheorem{maintheorem}[
    name={Theorem},
    refname={Theorem,Theorems},
    Refname={Theorem,Theorems},
]

\usepackage{geometry}
\usepackage{orcidlink}
\usepackage[foot]{amsaddr}
\newcommand\blfootnote[1]{%
  \begingroup
  \renewcommand\thefootnote{}\footnote{#1}%
  \addtocounter{footnote}{-1}%
  \endgroup
}

\title{Self-adjoint extensions of $k$-photon light--matter Hamiltonians}
\hypersetup{pdftitle={Self-adjoint extensions of k-photon light--matter Hamiltonians}}

\author{Felix Fischer\textsuperscript{\,1,*}\hspace{2pt}\orcidlink{0009-0001-3040-8184}\hspace{1pt}}

\author{Felix Knapp\textsuperscript{\,1}\hspace{2pt}\orcidlink{0009-0000-1419-1440}\hspace{1pt}}

\author{Daniel Burgarth\textsuperscript{\,1}\hspace{2pt}\orcidlink{0000-0003-4063-1264}\hspace{1pt}}

\author{Davide Lonigro\textsuperscript{\,1}\hspace{2pt}\orcidlink{0000-0002-0792-8122}\hspace{1pt}}

\address{\footnotesize \textsuperscript{1}Institute of Theoretical Physics, Friedrich-Alexander-Universität Erlangen-Nürnberg, Staudtstraße 7, 91058 Erlangen, Germany}

\renewcommand*{\thefootnote}{\fnsymbol{footnote}}

\begin{document}
\begin{abstract}
    Multiphoton light--matter interactions, in which a bosonic mode exchanges $k$ excitations at a time with a quantum system, are a source of genuine nonlinearity in quantum optics and are increasingly accessible experimentally. Here we study the class of operators $H = \hmat\otimes\id + \id\otimes\omega a^\ast a + \Sigma\otimes(a^\ast)^k + \Sigma^\ast\otimes a^k$ on $\hilbert\otimes L^2(\rnum)$, coupling a single bosonic mode to an arbitrary matter system through a bounded operator $\Sigma$. When $\Sigma$ is normal and nonzero, we prove that $H$ is self-adjoint if and only if $k\leq2$; for $k\geq3$ we compute the deficiency indices, parametrise all self-adjoint extensions, and show that every extension has purely discrete spectrum whenever the matter system is finite-dimensional. Our analysis rests on a block Jacobi decomposition paired with a suitable unitary transformation depending on the polar decomposition of $\Sigma$. The normality of $\Sigma$ is optimal: a $k$-photon Jaynes--Cummings model, with non-normal coupling, remains self-adjoint for every $k$. We illustrate our results on the $k$-photon Rabi and Dicke models.
\end{abstract}

\maketitle
\thispagestyle{empty}

\footnotetext{*Corresponding author: \href{mailto:felix.o.fischer@fau.de}{\texttt{felix.o.fischer@fau.de}}.}

\vspace{-0.5cm}
\noindent%

\section{Introduction}\label{sec:intro}
The quantum Rabi model is the paradigmatic model for light--matter interactions. It describes the interaction of a single bosonic mode with a spin degree of freedom, modelled by a Hamiltonian on $\cnum^2\otimes L^2(\rnum)$ acting as
\begin{equation}
    \label{eq:rabi}
    H= \id_{\cnum^2}\otimes \omega a^\ast a + \frac{\Omega}{2} \sigma_z \otimes\id_{L^2(\rnum)} + g\,\sigma_x \otimes (a^\ast + a) \, , 
\end{equation}
where $\omega$ is the frequency of the light field, $a^\ast$ and $a$ are the creation and annihilation operators of said field, $\Omega$ is the energy splitting of the spin, $\sigma_x,\sigma_z$ are the Pauli matrices, and $g$ is the coupling constant. This operator, when initially defined on $\cnum^2\otimes\domain_0$, with $\domain_0$ the span of the eigenstates $\phi_n$ of the number operator $\mathcal{N}=a^\ast a$ (the Fock states), is essentially self-adjoint.

This Hamiltonian and related models, including the Jaynes--Cummings model, the Dicke model, or multimode generalisations are ubiquitous in the physical literature, cf.~\cite{xie-quantumrabimodel-2017,larson-jaynescummingsmodel-2024} and references therein. All such models share a common feature: their Hamiltonians have an interaction term linear in $a$ and $a^\ast$, describing the exchange of one photon at a time.
In light of recent experimental realisations of genuine multiphoton interactions~\cite{corona-thirdorderspontaneousparametric-2011,boutin-effecthigherordernonlinearities-2017,chang-observationthreephotonspontaneous-2020,menard-emissionphotonmultiplets-2022,bencheikh-demonstratingquantumproperties-2022,eriksson-universalcontrolbosonic-2024,gregory-foursixphotonstimulated-2026,bazavan-squeezingtrisqueezingquadsqueezing-2026,Deng26}, there is an increased need for a better theoretical understanding of $k$-photon generalisations of \cref{eq:rabi}, the simplest instance being the so-called $k$-photon Rabi model:
\begin{equation}
    \label{eq:k-photon-rabi}
        H = \id_{\cnum^2}\otimes \omega a^\ast a + \frac{\Omega}{2} \sigma_z \otimes\id_{L^2(\rnum)}+ g\,\sigma_x\otimes \left((a^\ast)^k + a^k \right) \, .
\end{equation}
On the mathematical level, however, the presence of terms of order $k\geq3$ in $a,a^\ast$ makes the very question of essential self-adjointness of the operator nontrivial. This is exemplified by a series of recent results~\cite{ashhab-finitedimensionalapproximationsgeneralized-2026,ashhab-fractionalsqueezingspectra-2026,fischer-selfadjointrealizationshigherorder-2026,fischer-essentiallysingularlimits-2026} addressing the so-called higher-order squeezing operators, which correspond to specific polynomials in $a$ and $a^\ast$ of order higher than $2$, without any spin degree of freedom. These operators were shown not to be essentially self-adjoint on $\domain_0$, and to have deficiency indices equal to $(k,k)$.
In other words, additional information in the form of boundary conditions must be supplied in order to obtain well-defined dynamics.

It is then natural to ask whether a similar scenario arises for models describing $k$-photon light--matter interactions. A first result in this direction was obtained by Braak~\cite{braak-$k$photonquantumrabi-2025} for the specific case of the $k$-photon Rabi model.
Through a direct computation involving the Segal--Bargmann transform (which was famously used by the same author to compute the spectrum of the Rabi model~\cite{braak-integrabilityrabimodel-2011}), he showed that the generalised eigenvalue equation of the $k$-photon Rabi model has square-integrable solutions for all complex eigenvalues, contradicting previous results from Zhang~\cite{zhang-2modekphotonquantum-2017}.
This approach relies on an explicit computation of the generalised eigenvectors via the Bargmann transform, and therefore it does not extend to more general models for which such a computation is out of reach. A more general method, capable of treating a wider class of $k$-photon interactions at once, is therefore called for.

In this paper, we consider the following class of Hamiltonians on $\hilbert\otimes L^2(\rnum)$, where $\hilbert$ is an arbitrary (potentially infinite-dimensional) Hilbert space representing the matter system
\begin{equation}
    \label{eq:op}
    H = \id_\hilbert\otimes\omega a^\ast a  +  \hmat\otimes \id_{L^2(\rnum)}  + \Sigma\otimes (a^\ast)^k  + \Sigma^\ast\otimes a^k,
\end{equation}
where $\hmat$ is a bounded self-adjoint operator representing the free matter system, and $\Sigma$ is a \textit{normal} ($\Sigma^\ast \Sigma=\Sigma \Sigma^\ast$) bounded operator modulating the light--matter interaction. We initially define the operator with domain $\hilbert\otimes\domain_0$, $\domain_0$ again being the span of Fock states (choosing the Schwartz space $\mathcal{S}(\rnum)$ would yield identical results, cf. \cref{rem:felixs_choice}), and then consider its closure. Our results can be summarised as follows:

\begin{maintheorem}[\cref{prop:ess-sa,thm:main-result,cor:main-result-finite-dim,prop:main-result-inverse}]
Let $H$ be the closure of the operator defined above on $\hilbert\otimes\domain_0$, with $\Sigma\in\bounded(\hilbert)$ being normal and such that the inverse of its restriction to $(\Ker\Sigma)^\perp$ is bounded. Then:
\begin{itemize}
    \item[(i)] For $k\leq2$, $H$ is self-adjoint;
    \item[(ii)] For $k\geq3$ and $\Sigma\neq0$, $H$ is not self-adjoint and its deficiency spaces are given by
    \begin{equation}
        \dspace(\pm \iu) \cong \bigoplus_{m = 0}^{k-1}(\Ker\Sigma)^\perp \, .
    \end{equation}
    In particular, if $\dim\hilbert=d<\infty$, its self-adjoint extensions are parametrised by $kd'\times kd'$ unitary matrices, where $d'=\dim(\Ker\Sigma)^\perp$, and they all have purely discrete spectrum.
\end{itemize}
\end{maintheorem}
These results hold, in particular, whenever the matter space $\hilbert$ is finite-dimensional, provided that $\Sigma$ is normal and nonzero. This includes, in particular, the $k$-photon quantum Rabi model (\cref{ex:krabi}) and its generalisation to an arbitrary number of spins (\cref{ex:kdicke}). Importantly, the assumption that $\Sigma$ be normal \textit{cannot} be relaxed: the so-called $k$-photon Jaynes--Cummings model (\cref{ex:kjc}) is self-adjoint even for $k\geq3$, because of the presence of a fundamental symmetry---a conserved number of excitations---analogous to the one already exhibited by the standard Jaynes--Cummings model.

The key structural observation underlying our results is that all operators of the form \cref{eq:op} decompose as a direct sum of finitely many \textit{block Jacobi operators} (\cref{lem:decomposition}), following the approach initiated in~\cite{fischer-essentiallysingularlimits-2026}. For $k\leq 2$, self-adjointness then follows from the classical Carleman criterion. For $k\geq 3$, a further step is needed: a unitary transformation, constructed from the polar decomposition of $\Sigma$, brings each Jacobi component into a form whose spectral properties can be read off from the general theory of block Jacobi operators~\cite{swiderski-spectralpropertiesblock-2018}. It is in this step that the normality of $\Sigma$ enters in an essential way (also cf. \cref{rem:finite-dim-argument}); the Jaynes--Cummings counterexample (\cref{ex:kjc}) confirms that this assumption cannot be removed. The case where $\Sigma$ does not have a bounded inverse is then handled by a reduction argument.

This decomposition strategy was introduced in~\cite{fischer-essentiallysingularlimits-2026} to study higher-order squeezing operators. There, however, the operators decompose into \textit{scalar} Jacobi operators, to which \'Swiderski's theorem applies directly, with no further work required. The presence of a matter system is what makes the present problem genuinely harder: the decomposition now produces \textit{block} Jacobi operators whose coefficients do not satisfy the theorem's hypotheses as they stand, and it is precisely to remedy this that the unitary transformation above is needed.

This approach dispenses with the model-specific techniques previously employed for related problems, such as Birkhoff-type asymptotics of recurrence relations~\cite{fischer-selfadjointrealizationshigherorder-2026} or the Bargmann-space formalism~\cite{braak-integrabilityrabimodel-2011,braak-$k$photonquantumrabi-2025}: no detailed knowledge of the interaction operator $\Sigma$ is required beyond normality.

More broadly, the connection between multiphoton Hamiltonians and block Jacobi operators opens a direct channel between the mature spectral theory of such operators and the physics of higher-order light--matter interactions. For instance, the subordinacy theory of Moszy\'nski~\cite{moszynski-barriernonsubordinacyabsolutely-2025} could provide sufficient conditions for the emergence of absolutely continuous spectrum, while eigenvalue asymptotics from the Jacobi literature could be transferred to multiphoton models. More ambitiously, the same decomposition strategy applies in principle to arbitrary polynomials in $a$ and $a^\ast$, including the multimode setting relevant to continuous-variable quantum computing. These directions will be pursued in future work.

Our work is structured as follows. \cref{sec:preliminaries} contains all required mathematical preliminaries on block Jacobi operators;
in \cref{sec:main-result} we state and prove our main results; in \cref{sec:examples} we apply our results to the $k$-photon quantum Rabi model and its multi-spin generalisation, and discuss the $k$-photon Jaynes--Cummings model as an illustrative counterexample where the assumptions of our theorems are not met.

\section{Preliminaries on block Jacobi operators}
\label{sec:preliminaries}

As anticipated, in order to investigate the self-adjointness (or lack thereof) of the class of operators \eqref{eq:op}, we will decompose them as a direct sum of block tridiagonal (Jacobi) operators, cf. \cref{lem:decomposition}. In this section we provide some results about the spectral properties of block Jacobi operators. We refer to~\cite{schmudgen-unboundedselfadjointoperators-2012,schmudgen-momentproblem-2017,teschl-jacobioperatorscompletely-1999,akhiezer-classicalmomentproblem-2020} for Jacobi operators, and~\cite[Chapter~VII.2]{berezanskii-expansioneigenfunctionsselfadjoint-1968} and~\cite{swiderski-spectralpropertiesblock-2018,swiderski-asymptoticzerosdistribution-2025,moszynski-barriernonsubordinacyabsolutely-2025,budyka-selfadjointnessdiscretenessspectrum-2020,budyka-deficiencyindicesdiscreteness-2022a,budyka-deficiencyindicesblock-2024,braeutigam-deficiencyindicesoperators-2019,schulz-baldes-rotationnumbersjacobi-2007} for block Jacobi operators with possibly unbounded coefficients.

Let $\hilbert$ be a nontrivial complex separable Hilbert space equipped with a scalar product $\braket{\cdot,\cdot}$ and the associated norm $\norm{\cdot}$.
We denote the operator norm on $\hilbert$ by $\norm{\cdot}_\infty$, and the space of all bounded operators on $\hilbert$ by $\bounded(\hilbert)$. Moreover, we denote by $\ell(\nnum,\hilbert)$ the vector space of all sequences $u = (u_n)_{n \in \nnum}$, with $u_n \in \hilbert$ for all $n \in \nnum$, and introduce the space of square-summable sequences
\begin{equation}
	\ell^2(\nnum,\hilbert) = \left\{ u\in \ell(\nnum,\hilbert) \, : \,  \sum_{n = 0}^\infty\norm{u_n}^2 < \infty \right\} \, ,
\end{equation}
which is a Hilbert space with respect to the scalar product
\begin{equation}
	\braket{u,v}_{\ell^2(\nnum,\hilbert)} = \sum_{n = 0}^\infty \braket{u_n,v_n} \, , \quad \norm{u}_{\ell^2(\nnum,\hilbert)} = \sqrt{\braket{u,u}_{\ell^2(\nnum,\hilbert)}}\, .
\end{equation}
Finally, we denote by $\ell_0(\nnum,\hilbert)$ the space of finitely supported sequences
\begin{equation}
	\ell_0(\nnum,\hilbert) = \left\{u \in \ell(\nnum,\hilbert) \, : \, \exists N \in \nnum \: \text{s.t.} \:  u_n = 0 \quad \forall n \geq N\right\}.
\end{equation}
It is well-known that $\ell_0(\nnum,\hilbert)$ is a dense subspace of $\ell^2(\nnum,\hilbert)$.

\begin{definition}[Block Jacobi operator]
	\label{def:block-jacobi}
	Fix two sequences $(A_n)_{n \in \nnum}$, $(B_n)_{n \in \nnum}\subset\mathcal{B}(\hilbert)$ such that, for every $n \in \nnum$, $A_n$ has bounded inverse and $B_n$ is self-adjoint. We define the operator $\Jop$ on $\ell(\nnum,\hilbert)$ by
	\begin{align}
		(\Jop u)_n & = A_n u_{n+1} + B_n u_n + A_{n-1}^\ast u_{n-1}\,, \quad n \geq 1\,, \\
		(\Jop u)_0 & = A_0 u_1 + B_0 u_0.
	\end{align}
	Moreover, let $\Jmin$ be its restriction to $\ell_0(\nnum,\hilbert)$, and $J = \overline{\Jmin}$ its closure.
\end{definition}
One can see that $J\in\bounded(\ell^2(\nnum,\hilbert))$ if and only if the sequences $(A_n)_{n\in\nnum},(B_n)_{n\in\nnum}$ are bounded in the operator norm, that is, $\sup_{n\in\nnum}\|A_n\|_\infty<\infty$ and $\sup_{n\in\nnum}\|B_n\|_\infty<\infty$. For the purpose of this paper we will be mainly interested in the case where the two sequences above are unbounded, and so is $J$.
We remark that the action of a block Jacobi operator can be represented by an infinite block tridiagonal matrix:
    \begin{equation}
        \Jop = \begin{pmatrix}
                B_0 & A_0 & 0 & 0 & 0 & \dots \\
                A_0^\ast & B_1 & A_1 & 0 & 0 & \dots \\
                0 & A_1^\ast & B_2 & A_2 & 0 & \dots \\
                0 & 0 & A_2^\ast & B_3 & A_3 & \dots \\
                \vdots & \vdots & \vdots & \vdots & \vdots  & \ddots
               \end{pmatrix} 
    \end{equation}
We follow the same convention as in~\cite{fischer-essentiallysingularlimits-2026}, and respectively use calligraphic letters $\Jop$ and straight letters $J$ to denote the operator on the vector space $\ell(\nnum,\hilbert)$ and its restriction as a closed unbounded operator on the Hilbert space $\ell^2(\nnum,\hilbert)$. Note that $\Jmin$ is symmetric by construction, therefore so is its closure $J$.

Of particular relevance for block Jacobi operators is the \textit{generalised eigenvalue equation}: for $z \in \cnum$,
\begin{equation}
    \label{eq:generalized-eigenvalue}
    A_n u_{n + 1} + B_n u_n + A_{n-1}^\ast u_{n-1} = z u_n \quad \forall n \geq 1\, ,
\end{equation}
with $u\in\ell(\nnum,\hilbert)$. As this is a second-order recurrence relation and each $A_n$ has bounded inverse, specifying the initial conditions $u_0$ and $u_1$ is sufficient to determine a unique solution $u \in \ell(\nnum,\hilbert)$.
If the particular initial condition
\begin{equation}
    \label{eq:generalized-eigenvalue-initial}
    A_0 u_1 + B_0 u_0 = z u_0\, , 
\end{equation}
is satisfied, $u$ is a solution of the equation $\Jop u = z u$.

In the block Jacobi case, we also introduce an operator-valued generalised eigenvalue equation
\begin{equation}
    \label{eq:generalized-eigenvalue-matrix}
    A_n U_{n+1} + B_n U_n + A_{n-1}^\ast U_{n-1} = z U_n \quad \forall n \geq 1
\end{equation}
for $U \in \ell(\nnum,\bounded(\hilbert))$, the latter being the space of all $\bounded(\hilbert)$-valued sequences.
Again, specifying the initial conditions $U_0,U_1 \in \bounded(\hilbert)$ fixes the sequence $U$. Two specific solutions of \eqref{eq:generalized-eigenvalue-matrix}, analogous to the \textit{orthogonal polynomials} in the scalar case, are introduced in the following, cf.~\cite[p.~560]{berezanskii-expansioneigenfunctionsselfadjoint-1968}:
\begin{definition}[Orthogonal operator polynomials]
    \label{def:orthogonal-polynomials}
    The \textit{orthogonal operator polynomials of first and second kind} are the unique solutions $P(z) = (P_n(z))_{n \in \nnum}, Q(z) = (Q_n(z))_{n \in \nnum} \in \ell(\nnum,\bounded(\hilbert))$ of \cref{eq:generalized-eigenvalue-matrix} corresponding to the initial conditions
    \begin{alignat}{2}
        P_0(z) & = \id \quad & P_1(z) & = A_0^{-1}(z-B_0) \, ,  \\
        Q_0(z) & = 0 \quad & Q_1(z) & = A_0^{-1} \, .
    \end{alignat}
\end{definition}
It follows that $P_n(z)$ and $Q_{n+1}(z)$ are operator-valued polynomials of degree $n$ in $z$.
Moreover, $P(z)$ additionally satisfies the operator-valued version of the initial condition \cref{eq:generalized-eigenvalue-initial}, that is,
\begin{equation}
    \label{eq:generalized-eigenvalue-initial-matrix}
    A_0 U_1 + B_0 U_0 = z U_0\, , 
\end{equation}
whence 
\begin{equation}
    \Jop P(z) = z P(z).
\end{equation}

With these definitions at hand, we recall the following result on the deficiency subspaces $ \dspace(z) = \Ran(J-z)^\perp =  \Ker(J^\ast-z^\ast)$ of $J$ for nonreal $z$, cf.~\cite[Theorem VII.2.7]{berezanskii-expansioneigenfunctionsselfadjoint-1968} and~\cite{braeutigam-deficiencyindicesoperators-2019}:
\begin{lemma}
    \label{lem:deficiency-spaces}
    Let $J$ be the block Jacobi operator from \cref{def:block-jacobi}, and let $P(z)$ be the associated orthogonal polynomials of first kind from \cref{def:orthogonal-polynomials}.
    For every $z \in \cnum\setminus\rnum$ the deficiency subspaces of $J$ are given by
    \begin{equation}
        \dspace(z) = \left\{ (P_n(z^\ast) x)_{n \in \nnum} \, : \, x \in \hilbert, \;\sum_{n = 0}^\infty \norm{P_n(z^\ast) x}^2 < \infty \right\} \, .
    \end{equation}
    Moreover, the deficiency indices of $J$,
\begin{equation}
    n_+ = \dim \dspace(z)\,, \quad n_- = \dim \dspace(z^\ast) \quad \Imag z > 0 \,, 
\end{equation}
are independent of $z$. 
\end{lemma}
If $n_+ = n_-  = 0$, then $J$ is self-adjoint; if all solutions of \cref{eq:generalized-eigenvalue} are square-summable, we call $J$ \textit{maximally indeterminate}, and in particular, we obtain
\begin{equation}
    \dspace(z) = \left\{ (P_n(z^\ast) x)_{n \in \nnum} \, : \, x \in \hilbert \right\} \cong \hilbert \cong \dspace(z^\ast) \, .
\end{equation}

Next we recall a sufficient condition for the self-adjointness of $J$:
\begin{proposition}
    \label{prop:carleman}
    Let $J$ be the block Jacobi operator from \cref{def:block-jacobi}, and assume that the Carleman condition
    \begin{equation}
        \label{eq:carleman}
        \sum_{n = 0}^\infty \frac{1}{\norm{A_n}_\infty} = \infty
    \end{equation}
    is satisfied.
    Then $J$ is self-adjoint~\cite[Theorem VII.2.9]{berezanskii-expansioneigenfunctionsselfadjoint-1968}.
\end{proposition}
In the other direction, we will state a sufficient condition for $J$ to be maximally indeterminate ($\dspace(z)\cong\hilbert$) due to \'Swiderski~\cite{swiderski-spectralpropertiesblock-2018}, which will be central in our analysis in \cref{sec:main-result}. To this purpose, we begin by recalling the following definition:
\begin{definition}[Total variation]
	\label{def:variation}
	Let $X \in \ell(\nnum,\bounded(\hilbert))$. The \textit{total variation} $\variation(X)$ of $X$ is defined by
	\begin{equation}
	    \variation(X) = \sum_{n=0}^\infty \norm{X_{n+1}-X_n}_\infty \,.
	\end{equation}
\end{definition}
\begin{theorem}[\cite{swiderski-spectralpropertiesblock-2018}]
    \label{thm:swiderski}
    Let $J$ be the block Jacobi operator from \cref{def:block-jacobi}. Assume the following additional conditions on the sequences $(A_n)_{n\in\nnum},(B_n)_{n\in\nnum}$:
       \begin{enumerate}[(i)]
        \item $ \variation\left((A_n^{-1})_{n \in \nnum}\right) + \variation\left((A_n^{-1} B_n)_{n \in \nnum}\right) + \variation\left((A_n^{-1} A_{n-1}^\ast)_{n \in \nnum}\right) < \infty$;
        \item $\sum_{n=0}^\infty \norm{A_n}_\infty^{-1} < \infty$, that is, the Carleman condition does not hold;
        \item The operator-norm limits 
        \begin{equation}
            T = \lim_{n \to \infty}A_n^{-1}, \quad Q = \lim_{n \to \infty} A_n^{-1}B_n, \quad R = \lim_{n \to \infty} A_n^{-1}A_{n-1}^\ast,\quad C = \lim_{n \to \infty} \norm{A_n}_\infty^{-1}A_n 
        \end{equation}
        exist in $\bounded(\hilbert)$; in particular, $C$ is invertible.
        \item There exists $z \in \cnum$ such that the quadratic form associated with the operator
        \begin{equation}
	\mathcal{F}(z) = \Real \left(
	\begin{pmatrix}
	 0 & -C \\ C & 0
\end{pmatrix}
\begin{pmatrix}
	0 & \id \\ -R & z T - Q
\end{pmatrix}
	\right),
\end{equation}
where $\Real A= \frac{1}{2}(A+A^\ast)$, is either strictly positive or strictly negative.
\end{enumerate}
Then, for all $z \in \cnum$, $\dspace(z)\cong\hilbert$. In particular, $J$ is not self-adjoint, but maximally indeterminate.
\end{theorem}

\section{Main Results}
\label{sec:main-result}

Throughout this work, we consider operators describing the interaction between a single bosonic mode and a matter system.
The matter system is represented by a (potentially infinite-dimensional) Hilbert space $\hilbert$; the bosonic mode is represented by the Hilbert space $L^2(\rnum)$ of square-integrable functions on the real line, according to the following standard rules. Given the sequence $(\phi_n)_{n\in\nnum}\subset L^2(\rnum)$ of normalised Hermite functions, one defines the operators $a,a^\ast$ on $\mspan(\phi_n)_{n\in\nnum}$ via
\begin{align}
    a \phi_0 & = 0 \, , \\
    a \phi_n & = \sqrt{n}\phi_{n-1} \quad \forall n \geq 1 \, , \\
    a^\ast \phi_n & = \sqrt{n+1} \phi_{n+1} \, , 
\end{align}
and linear extension. Physically, $\phi_n$ is interpreted as the state of the boson field with $n$ excitations (the $n$th Fock state), and $a,a^\ast$ as the annihilation and creation operators of the field. Moreover, the operator $a^\ast a$ satisfies $a^\ast a\phi_n=n\phi_n$ and is thus the number operator of the field; this is an essentially self-adjoint operator. The energy of the free field is then described by $\omega a^\ast a$, with $\omega \geq 0$ representing the energy of a single excitation.

\begin{definition}
	\label{def:op}
	Let $k \in \nnum$, $\omega \geq 0$, and $\hmat, \Sigma \in \bounded(\hilbert)$, with $\hmat^\ast=\hmat$. We define the operator $\Hmin$ on $\hilbert\otimes L^2(\rnum)$ via
    \begin{equation}
    	\Hmin = \id_{\hilbert} \otimes\omega a^\ast a   + \hmat\otimes\id_{L^2(\rnum)} \, +        
        \Sigma\otimes (a^\ast)^k + \Sigma^\ast \otimes a^k 
    \end{equation}
    with domain $\hilbert\otimes\domain_0$, where $\domain_0 = \mspan(\phi_n)_{n \in \nnum}$. 
    
    Moreover, we denote by $H = \overline{\Hmin}$ the closure of $\Hmin$.
\end{definition}
Above, $\id_{\hilbert} \otimes\omega a^\ast a   + \hmat\otimes\id_{L^2(\rnum)}$ represents the free (decoupled) energy of the matter--boson pair; the additional term represents an interaction implementing the creation or annihilation of $k$ field excitations, analogously to the $k$-photon Rabi model introduced in \cref{sec:intro} (cf. \cref{ex:krabi}). The operator $\Sigma$ implements the specific structure of the light--matter interaction. We stress that the free matter Hamiltonian $\hmat$ is assumed bounded throughout; this is automatic for finite-dimensional matter systems, but excludes an unbounded free matter dynamics---for instance, a second bosonic mode in the role of the matter system---which would require a genuine extension of our methods.

\begin{remark}\label{rem:felixs_choice}
The choice $\domain_0=\mspan(\phi_n)_{n\in\nnum}$ is not the only possible one. One can see that all results throughout the paper would hold true if $\domain_0$ in the initial definition of $\Hmin$ were to be replaced by the space
    \begin{equation}
        \domain\left(\overline{(a^\ast a)^{m/2}}\right) = \left\{ \psi = \sum_{n = 0}^\infty c_n \phi_n \, : \, \sum_{n = 0}^\infty n^{m} |c_n|^2 < \infty \right\} \,, \quad m \geq k \geq 2
    \end{equation}
    or the Schwartz space 
    \begin{equation}
        \mathcal{S}(\rnum) = \left\{ \psi = \sum_{n = 0}^\infty c_n \phi_n \, : \, \sum_{n = 0}^\infty n^{m} |c_n|^2 < \infty \quad \forall m \geq 0 \right\} \, .
    \end{equation}
    This follows from the fact that, because of $(a^\ast)^k$ and $a^k$ being relatively bounded with respect to $(a^\ast a)^{m/2}$ for all $m \geq k$, the closures $H$ of $\Hmin$ with either of these initial domains coincide.
\end{remark}

Clearly, the operator $H$ is symmetric. In order to determine whether it is self-adjoint, we follow an approach analogous to the one in~\cite{fischer-essentiallysingularlimits-2026}. We expand states in $L^2(\rnum)$ in the Fock basis, hence using the standard unitary equivalence
\begin{equation}
    L^2(\rnum)\cong\ell^2(\nnum),
\end{equation}
whence
\begin{equation}
    \hilbert\otimes L^2(\rnum)\cong \hilbert\otimes \ell^2(\nnum)\cong \ell^2(\nnum,\hilbert).
\end{equation}
In this representation, we will see that $H$ decomposes as a direct sum of block tridiagonal (Jacobi) operators (cf. \cref{lem:decomposition}) as a direct consequence of the interaction term in $H$ only containing finite powers of $a,a^\ast$. This will enable us to apply results from the theory of block Jacobi operators (cf. \cref{sec:preliminaries}).

\subsection{Case 1: $\Sigma$ with bounded inverse}

We begin by restricting our discussion to the case where $\Sigma$ has a bounded inverse. As we will show later on (\cref{prop:main-result-inverse}), this assumption can be removed. Here and in the following, $(\cdot,\cdot)$ will denote the Pochhammer symbol:
\begin{equation}
    \label{eq:def-pochhammer}
    (x,s) = \prod_{i =  1}^s (x+i-1) = x (x+1) \dots (x+s-1) \, .
\end{equation} 

\begin{proposition}
    \label{lem:decomposition}
	Let $H$ be the operator from \cref{def:op}, and suppose that $\Sigma$ has a bounded inverse.
	Then,
	\begin{equation}
	    \label{eq:decomposition}
		H \simeq \bigoplus_{m = 0}^{k-1} J^{(m)} \, ,
    \end{equation}
    where $J^{(m)}$ is the block Jacobi operator on $\ell^2(\nnum,\hilbert)$ (\cref{def:block-jacobi}) corresponding to the sequences $(A^{(m)}_r)_{r\in\nnum},(B^{(m)}_r)_{r\in\nnum}\subset\mathcal{B}(\hilbert)$ given by
    \begin{align}\label{eq:coefficients1}
    	A^{(m)}_r &= \beta_{m+rk} \aint^\ast, \, \\ \label{eq:coefficients2}B^{(m)}_r &= \omega (m + rk) \id_\hilbert + \hmat,
    \end{align}
with $\beta_n = \sqrt{(n+1,k)}$.
\end{proposition}
\begin{proof}
Given $u\in\hilbert$, a calculation analogous to the one in~\cite[Lemmas~2.5--2.6]{fischer-selfadjointrealizationshigherorder-2026} yields
    \begin{equation}
        H ( u\otimes \phi_n ) = \beta_n  (\Sigma u )\otimes \phi_{n+k}
        + \omega n  u \otimes \phi_n + (\hmat u)\otimes \phi_n  
        + \beta_{n-k} (\Sigma^\ast u)\otimes  \phi_{n-k},
    \end{equation}
    where
    \begin{equation}
        \beta_n = \sqrt{(n+1,k)} =  \sqrt{(n+1)\dots(n+k)} \, ,
    \end{equation}
   with the convention $\phi_{n} = 0$ for $n < 0$.
    Thus, for every $0 \leq m \leq k-1$, $H$ leaves the subspace $\hilbert\otimes\mspan(\phi_{m + rk})_{r \in \nnum}$ invariant, and its restriction to this subspace reads
\begin{equation}
        H (u\otimes \phi_{m + rk}) = ((A^{(m)}_r)^\ast u)\otimes\phi_{m + (r+1)k}  +  (B^{(m)}_r u)\otimes \phi_{m + rk} +  (A^{(m)}_{r-1} u)\otimes \phi_{m + (r-1)k}\,,
    \end{equation}
    where $A_r^{(m)} = \beta_{m + rk} \aint^\ast$ and $B_r^{(m)} = \omega (m + rk) \id_\hilbert + \hmat$. Since $\aint$ has a bounded inverse, so do $\aint^\ast$ and $A_r^{(m)}$.    
    Using the standard unitary equivalence between $\hilbert\otimes\mspan(\phi_{m + rk})_{r \in \nnum}$ and $\ell_0(\nnum,\hilbert)$, this proves that the restriction of $H$ to $\hilbert\otimes\mspan(\phi_{m + rk})_{r \in \nnum}$ is unitarily equivalent to the block Jacobi operator restricted to finitely supported sequences, $J^{(m)}_{\mathrm{min}}$ (\cref{def:block-jacobi}). Therefore, $\Hmin$ is unitarily equivalent to $\bigoplus_{m=0}^{k-1}J^{(m)}_{\mathrm{min}}$, and the claimed property follows by taking closures~\cite[p.~79]{teschl-mathematicalmethodsquantum-2009}.
\end{proof}

As anticipated, the decomposition of $H$ into block Jacobi operators enables us to apply results from \cref{sec:preliminaries}. We begin with the Carleman condition, cf. \cref{prop:carleman}, which will readily enable us to prove self-adjointness of $H$ for $k\leq2$. 
\begin{lemma}
    \label{lem:our-carleman-beta}
	For all $0 \leq m \leq k-1$,
	\begin{equation}
	    \sum_{n = 0}^\infty \frac{1}{\beta_{m + nk}} = \infty
	\end{equation}
	if and only if $k \leq 2$.
\end{lemma}
\begin{proof}
This follows from a straightforward computation:    \begin{align}
        \beta_{m + nk} & = (m+nk+1,k)^{1/2} = \prod_{l = 1}^k (m+nk+l)^{1/2} \\
        & = (nk)^{k/2} \prod_{l = 1}^k \left(1 + \frac{m+l}{nk}\right)^{1/2} = (nk)^{k/2} \left(1 + O(1/n)\right)\, ,
    \end{align}
    and hence
    \begin{equation}
        \sum_{n = 0}^\infty \frac{1}{\beta_{m + nk}} = \frac{1}{k^{k/2}} \sum_{n = 1}^\infty \frac{1}{n^{k/2} (1+O(1/n))} \, ,
    \end{equation}
    which is finite if and only if $k/2 >1$.
\end{proof}
Above, we adopted the usual big-O notation, cf.~\cite[\S 4]{olver-asymptoticsspecialfunctions-2010}:
given two sequences $(x_n)_{n \in \nnum}$ and $(y_n)_{n \in \nnum}$ of complex numbers, we write $x_n = O(y_n)$ if there exist constants $C> 0$ and $N \in \nnum$ such that $|x_n| \leq C |y_n|$ for all $n \geq N$.
\begin{proposition}
    \label{prop:ess-sa}
	Let $H$ be the operator from \cref{def:op}, and assume that $\Sigma$ has a bounded inverse. Assume $k\leq2$.
	Then $H$ is self-adjoint.
\end{proposition}
\begin{proof}
	By \cref{lem:decomposition}, $H\simeq\bigoplus_{m=0}^{k-1}J^{(m)}$, with $J^{(m)}$ being a block Jacobi operator with coefficients $A^{(m)}_r,B^{(m)}_r$ from \cref{eq:coefficients1,eq:coefficients2}. In particular, $\norm*{A^{(m)}_n}_\infty = \beta_{m + nk} \norm{\aint^\ast}_\infty$, and \cref{lem:our-carleman-beta} implies
	\begin{equation}
	    \sum_{n = 0}^\infty \frac{1}{\norm*{A^{(m)}_n}_\infty} = \infty \, ;
	\end{equation}
	hence $J^{(m)}$ is self-adjoint by the Carleman condition (\cref{prop:carleman}), and by \cref{lem:decomposition} so is $H$.
\end{proof}

The case $k \geq 3$ will require a more careful analysis. Our strategy is to utilise \cref{thm:swiderski} to prove that the operators $J^{(m)}$, and thus $H$, are not self-adjoint and admit multiple self-adjoint extensions. However, the operators $J^{(m)}$ themselves do not have, in general, strictly positive or negative off-diagonal entries, and thus fail to satisfy the assumptions of \cref{thm:swiderski}. A further unitary transformation will be needed to circumvent this issue. 

This is precisely the point where we will need $\Sigma$ to be \textit{normal}, that is, $\Sigma^\ast\Sigma=\Sigma\Sigma^\ast$ (also cf. \cref{rem:finite-dim-argument}). Consequently, it admits a \textit{polar decomposition}: there exist two commuting operators $U,\pint\in\mathcal{B}(\hilbert)$, with $U$ being unitary and $\pint\geq0$, such that
\begin{equation}\label{eq:polar}
    \Sigma=U\pint;
\end{equation}
moreover, if $\Sigma$ has a bounded inverse, then so does $\pint$~\cite[Theorem 12.35]{rudin-functionalanalysis-2007}.

\begin{lemma}
	\label{lem:trafo}
	Let $\Sigma\in\bounded(\hilbert)$ be a normal operator with bounded inverse, and consider its polar decomposition $\Sigma=U\pint$, with $U$ unitary and $\pint>0$.
	Define the operator $\Uop$ on $\ell^2(\nnum,\hilbert)$ by
	\begin{equation}
	    (\Uop u)_n = (U^\ast)^n u_n \quad \forall u \in \ell^2(\nnum,\hilbert) .
	\end{equation}
Then the following properties hold:
\begin{itemize}
    \item [(i)] $\Uop$ is unitary and leaves $\ell_0(\nnum,\hilbert)$ invariant;
    \item [(ii)] Let $k\geq 3$, $0 \leq m \leq k-1$, and $J^{(m)}$ be the block Jacobi operator from \cref{lem:decomposition}. Define $\tilde{J}^{(m)}: = \mathcal{U} J^{(m)} \mathcal{U}^\ast$.
	Then $\tilde{J}^{(m)}$ is a block Jacobi operator on $\ell^2(\nnum,\hilbert)$ with sequences $(\tilde{A}^{(m)}_n)_{n\in\nnum},(\tilde{B}^{(m)}_n)_{n\in\nnum} \subset\bounded(\hilbert)$ given by
	\begin{align}\label{eq:coefficients1tilde}
		\tilde{A}^{(m)}_n  & = \beta_{m + nk} \pint \, , \\\label{eq:coefficients2tilde}
		\tilde{B}^{(m)}_n & = \omega(m+nk)\id_\hilbert + (U^\ast)^n \hmat U^n  \, .
	\end{align}
	In particular, for all $n$, $\tilde{A}^{(m)}_n$ has bounded inverse and $\tilde{B}^{(m)}_n$ is self-adjoint.
\end{itemize}
\end{lemma}

\begin{proof}

(i): As $U$ is unitary, so is $(U^\ast)^n$ for every $n\in\nnum$, which readily implies unitarity of $\Uop$. Invariance of $\ell_0(\nnum,\hilbert)$ is immediate.

(ii): It suffices to prove that the action of $\tilde{J}^{(m)}$ on $\ell_0(\nnum,\hilbert)$ coincides with the one of the block Jacobi operator with coefficients $\tilde{A}^{(m)}_n$, $\tilde{B}^{(m)}_n$. 

Let $u \in \ell_0(\nnum,\hilbert)$. We have
	\begin{align}
	   		(\tilde{J}^{(m)}u)_n & = (\mathcal{U} J^{(m)} \mathcal{U}^\ast u)_n = (U^\ast)^n (J^{(m)} \mathcal{U}^\ast u)_n \\
		& = (U^\ast)^n A^{(m)}_n U ^{n+1} u_{n+1} + (U^\ast)^n B^{(m)}_n  U^{n} u_n + (U^\ast)^n(A^{(m)}_{n-1})^\ast U^{n-1} u_{n-1} \, . \label{proofeq:tilde-j}
	\end{align}
	We use the polar decomposition $\aint = U \pint$ and $\pint^\ast = \pint$. Since $U$ and $\pint$ commute,
	\begin{align}
		(U^\ast)^n A^{(m)}_n U^{n+1} & = \beta_{m + nk} (U^\ast)^n U^\ast \pint U^{n + 1} = \beta_{m + nk} \pint =\tilde{A}^{(m)}_n\, , \\
		(U^\ast)^n(A^{(m)}_{n-1})^\ast U^{n-1} & = \beta_{m + (n-1)k} (U^\ast)^n \left(U^\ast \pint\right)^\ast U^{n -1} = \beta_{m + (n-1)k} \pint \, =\left(\tilde{A}^{(m)}_{n-1}\right)^\ast,
	\end{align}
and	\begin{equation}
	    (U^\ast)^n B^{(m)}_n  U^n = (U^\ast)^n \left( \omega(m + nk) \id_\hilbert + \hmat \right) U^n = \omega(m + nk) \id_\hilbert + (U^\ast)^n \hmat U^n = \tilde{B}^{(m)}_n\, ,
	\end{equation}
	whence
	\begin{equation}
	    (\tilde{J}^{(m)}u)_n = \tilde{A}^{(m)}_n u_{n+1} + \tilde{B}^{(m)}_n u_n + \left(\tilde{A}^{(m)}_{n-1}\right)^\ast u_{n-1}.
	\end{equation}
     This proves the claimed equality. Clearly $\tilde{B}^{(m)}_n$ is self-adjoint; since $\aint$ has a bounded inverse, so does $\pint$, and thus $\tilde{A}^{(m)}_n$.
\end{proof}
The following lemma shows that the sequences $(\tilde{A}^{(m)}_n)_{n \in \nnum}$ and $(\tilde{B}^{(m)}_n)_{n \in \nnum}$ fulfil assumption (i) of \cref{thm:swiderski}.
\begin{lemma}
    \label{lem:variation}
Let $\Sigma\in\bounded(\hilbert)$ be a normal operator with bounded inverse, and polar decomposition $\Sigma=U\pint$; let $k\geq 3$, $0 \leq m \leq k-1$, and $\tilde{J}^{(m)}$ be the block Jacobi operator with coefficients $\tilde{A}^{(m)}_n$ and $\tilde{B}^{(m)}_n$ from \cref{lem:trafo} (\cref{eq:coefficients1tilde,eq:coefficients2tilde}). Then
	\begin{align}
		\mathcal{V}((\tilde{A}^{(m)}_n)^{-1}) &< \infty \, , \label{proofeq:var1}\\ 		
        \mathcal{V}((\tilde{A}^{(m)}_n)^{-1} \tilde{B}^{(m)}_n) &< \infty \, , \label{proofeq:var2}\\ 
        \mathcal{V}((\tilde{A}^{(m)}_n)^{-1}(\tilde{A}^{(m)}_{n-1})^\ast) &< \infty \label{proofeq:var3}
	\end{align}
    with $\mathcal{V}(\cdot)$ denoting the total variation (\cref{def:variation}).
\end{lemma}
\begin{proof}
The statement is well-posed since, as shown in \cref{lem:trafo}, $\tilde{A}^{(m)}_n$ has a bounded inverse for all $n$. 
We begin with \cref{proofeq:var1}. We have
	\begin{equation}
		\mathcal{V}((\tilde{A}^{(m)}_n)^{-1})  = \sum_{n = 1}^\infty \norm{(\tilde{A}^{(m)}_n)^{-1} - (\tilde{A}^{(m)}_{n-1})^{-1}}_\infty \leq \sum_{n = 1}^\infty \left(\norm{(\tilde{A}^{(m)}_n)^{-1}}_\infty + \norm{ (\tilde{A}^{(m)}_{n-1})^{-1}}_\infty\right) ,
	\end{equation}
and
\begin{equation}
	    \norm{(\tilde{A}_n^{(m)})^{-1}}_\infty = \norm{\beta_{m+nk}^{-1}\pint^{-1}}_\infty = \beta_{m + nk}^{-1} \norm{\pint^{-1}}_\infty \, .
	\end{equation}
	Using these estimates and recalling that the Carleman condition is not met for $k\geq3$ (\cref{lem:our-carleman-beta}), we get
	\begin{equation}
	    \mathcal{V}((\tilde{A}^{(m)}_n)^{-1}) \leq  \sum_{n = 1}^\infty \left(\norm{(\tilde{A}^{(m)}_n)^{-1}}_\infty + \norm{ (\tilde{A}^{(m)}_{n-1})^{-1}}_\infty\right) \leq 2 \norm{\pint^{-1}}_\infty \sum_{n = 0}^\infty \beta_{m + nk}^{-1} < \infty.
	\end{equation}
	To prove \cref{proofeq:var2,proofeq:var3}, we will need the following asymptotic expansion:
    \begin{equation}
	    \label{proofeq:beta-inverse}
	    \beta_{m + nk}^{-1} = (m + nk)^{-k/2} \left( 1 + O(1/n)\right),
	\end{equation}
    which follows directly from the definition of $\beta_n$, cf. \cref{lem:decomposition}:
	\begin{align}
	    \label{proofeq:beta-expansion}
	    \beta_{m+nk} & = \prod_{l = 1}^k (m + nk + l)^{1/2} = (m + nk)^{k/2} \prod_{l = 1}^k \left(1 + \frac{l}{m + nk}\right)^{1/2} \\
					& = (m + nk)^{k/2} \left( 1 + \frac{c_k}{n} + O(1/n^2) \right) \,  ,
	\end{align}
	with some $c_k \in \rnum$.	

	We begin with \cref{proofeq:var2}. We have 
	\begin{equation}
		(\tilde{A}^{(m)}_n)^{-1} \tilde{B}^{(m)}_n = \frac{\omega(m+nk)}{\beta_{m+nk}}\pint^{-1} + \beta_{m+nk}^{-1}\pint^{-1} (U^\ast)^n \hmat U^n \, .
	\end{equation}
	By using the asymptotic expansion \eqref{proofeq:beta-inverse},
	\begin{align}
	    \label{proofeq:var-ab-expansion}
		& \norm{(\tilde{A}^{(m)}_n)^{-1} \tilde{B}^{(m)}_n - (\tilde{A}^{(m)}_{n-1})^{-1} \tilde{B}^{(m)}_{n-1}}_\infty  \leq \left| \frac{\omega(m+nk)}{\beta_{m+nk}} - \frac{\omega(m+(n-1)k)}{\beta_{m+(n-1)k}} \right| \norm{\pint^{-1}}_\infty \\
		& \qquad+ \norm{\beta_{m+nk}^{-1}\pint^{-1} (U^\ast)^n \hmat U^n - \beta_{m+(n-1)k}^{-1}\pint^{-1} (U^\ast)^{n-1} \hmat U^{n-1}}_\infty \\
	  &\quad  \leq \omega \left| (m + nk)^{1-k/2}\left(1 + O(1/n)\right) - (m + (n-1)k)^{1-k/2}\left(1 + O(1/n)\right)\right| \norm{\pint^{-1}}_\infty \\
		& \qquad + 2 \beta_{m+(n-1)k}^{-1} \norm{\pint^{-1}}_\infty\norm{\hmat}_\infty \, ,
	\end{align}
	where we used $\beta_{m + nk}^{-1} \leq \beta_{m + (n-1)k}^{-1}$.
	Moreover,
	\begin{equation}
	    (m + (n-1)k)^{1-k/2} = (m + nk)^{1-k/2} \left(1 - \frac{k}{m + nk}\right)^{1-k/2} = (m + nk)^{1-k/2} \left( 1+ O(1/n)\right)
	\end{equation}
	and hence
	\begin{multline}
	    (m + nk)^{1-k/2}\left(1 + O(1/n)\right) - (m + (n-1)k)^{1-k/2}\left(1 + O(1/n)\right)  \\
					= (m + nk)^{1-k/2}\left( 1  - \left( 1 + O(1/n)\right)\left(1 + O(1/n)\right)\right) \\
					= (m + nk)^{1-k/2} O(1/n) = O(n^{-k/2}) \, .
	\end{multline}
	Combining this with \cref{proofeq:var-ab-expansion}, we get
	\begin{equation}
	    \norm{(\tilde{A}^{(m)}_n)^{-1} \tilde{B}^{(m)}_n - (\tilde{A}^{(m)}_{n-1})^{-1} \tilde{B}^{(m)}_{n-1}}_\infty \leq \omega \norm{\pint^{-1}}_\infty O(n^{-k/2}) + 2 \beta_{m+(n-1)k}^{-1} \norm{\pint^{-1}}_\infty\norm{\hmat}_\infty\, ,
	\end{equation}
	and, again using the failure of the Carleman condition for $k\geq3$ (\cref{lem:our-carleman-beta}), we obtain
	\begin{align}
		\mathcal{V}((\tilde{A}^{(m)}_n)^{-1} \tilde{B}^{(m)}_n) & = \sum_{n = 1}^\infty \norm{(\tilde{A}^{(m)}_n)^{-1} \tilde{B}^{(m)}_n - (\tilde{A}^{(m)}_{n-1})^{-1} \tilde{B}^{(m)}_{n-1}}_\infty \\
		& \leq \omega \norm{\pint^{-1}}_\infty \sum_{n = 1}^\infty O(n^{-k/2}) + 2 \norm{\pint^{-1}}_\infty\norm{\hmat}_\infty \sum_{n = 0}^\infty \beta_{m + nk}^{-1} < \infty \, .
    \end{align}

We finally prove \cref{proofeq:var3}. As $\pint^\ast = \pint$,
	\begin{equation}
	    \label{proofeq:beta-ratio}
		(\tilde{A}^{(m)}_n)^{-1}(\tilde{A}^{(m)}_{n-1})^\ast  = \frac{\beta_{m + (n-1)k}}{\beta_{m+nk}} \pint^{-1} \pint^\ast = \frac{\beta_{m + (n-1)k}}{\beta_{m+nk}} \id_\hilbert \, ,
    \end{equation}
    and using \cref{proofeq:beta-expansion}, we get
    \begin{align}
        \frac{\beta_{m + (n-1)k}}{\beta_{m+nk}} & = \left(\frac{m + (n-1)k}{m + nk}\right)^{k/2}\frac{ 1 + c_k/(n-1) + O(1/n^2)}{1+c_k/n + O(1/n^2)} \\
        & = \left(1 - \frac{k}{m + nk}\right)^{k/2} \left(1 + \frac{c_k}{n-1} - \frac{c_k}{n} + O(1/n^2)\right) \\
        & = \left(1 - \frac{k^2}{2 (m + nk)} + O(1/n^2)\right)\left( 1+ O(1/n^2)\right) = 1 - \frac{k^2}{2 (m + nk)} + O(1/n^2)
    \end{align}
    and hence
    \begin{equation}
        \frac{\beta_{m + (n-1)k}}{\beta_{m+nk}} - \frac{\beta_{m + (n-2)k}}{\beta_{m+(n-1)k}} = -\frac{k^2}{2 (m + nk)} + \frac{k^2}{2 (m + (n-1)k)} + O(1/n^2) = O(1/n^2) \, ,
    \end{equation}
    whence, using \cref{proofeq:beta-ratio}, 
    \begin{equation}
        \mathcal{V}((\tilde{A}^{(m)}_n)^{-1}(\tilde{A}^{(m)}_{n-1})^\ast) = \sum_{n = 1}^\infty \left|\frac{\beta_{m + (n-1)k}}{\beta_{m+nk}}-\frac{\beta_{m + (n-2)k}}{\beta_{m+(n-1)k}}\right| < \infty\,,
    \end{equation}
    which concludes the proof.
\end{proof}

With these results at hand, we can finally prove that each operator $\tilde{J}^{(m)}$ is maximally indeterminate, i.e. its deficiency subspaces are unitarily equivalent to $\hilbert$, by applying \cref{thm:swiderski}.

\begin{proposition}
    \label{prop:deficiency}
  Let $\Sigma\in\bounded(\hilbert)$ be a normal operator with bounded inverse, and polar decomposition $\Sigma=U\Pi$; let $k\geq 3$, $0 \leq m \leq k-1$, and $\tilde{J}^{(m)}$ be the block Jacobi operator with coefficients $\tilde{A}^{(m)}_n$ and $\tilde{B}^{(m)}_n$ from \cref{lem:trafo}.
    Then $\tilde{J}^{(m)}$ is maximally indeterminate, that is,
    \begin{equation}
        \Ker\left(\left( \tilde{J}^{(m)}\right)^\ast - z \right)  \cong \hilbert \, 
    \end{equation}
    for all $z\in\cnum$.
\end{proposition}
\begin{proof}
	We will check that conditions (i)--(iv) in \cref{thm:swiderski} are satisfied by $\tilde{J}^{(m)}$. Recall that the coefficients $\tilde{A}^{(m)}_n,\tilde{B}^{(m)}_n$ are given by \cref{eq:coefficients1tilde,eq:coefficients2tilde}. Condition (i) is precisely \cref{lem:variation}; besides, condition (ii) follows directly from \cref{lem:our-carleman-beta}:
	\begin{equation}
	    \sum_{n = 0}^\infty \frac{1}{\norm{\tilde{A}^{(m)}_n}_\infty} = \frac{1}{\norm{\pint}_\infty} \sum_{n = 0}^\infty \beta_{m + nk}^{-1} < \infty.
	\end{equation}
    We check condition (iii). We have
	\begin{align}
		\lim_{n \to \infty}\norm{(\tilde{A}^{(m)}_n)^{-1}}_\infty   = \lim_{n \to \infty} \beta_{m + nk}^{-1} \norm{\pint^{-1}}_\infty &= 0 \,, \\
		\lim_{n \to \infty} \norm{(\tilde{A}^{(m)}_n)^{-1} \tilde{B}^{(m)}_n}_\infty  \leq \lim_{n \to \infty} \omega\frac{m+nk}{\beta_{m+nk}} \norm{\pint^{-1}}_\infty + \lim_{n \to \infty} \beta_{m + nk}^{-1} \norm{\pint^{-1}}_\infty \norm{\hmat}_\infty & = 0 \, ,\\
		\lim_{n \to \infty} \norm{(\tilde{A}^{(m)}_n)^{-1} (\tilde{A}^{(m)}_{n-1})^\ast - \id_\hilbert}_\infty  = \lim_{n \to \infty} \left| \frac{\beta_{m+(n-1)k}}{\beta_{m+nk}}-1\right|\norm{\id_\hilbert}_\infty & = 0 \,, \\
		\lim_{n \to \infty} \norm{ \tilde{A}^{(m)}_n \norm{\tilde{A}^{(m)}_n}^{-1} - \pint \norm{\pint}^{-1}_\infty}_\infty = \lim_{n \to \infty} \norm{\pint}^{-1} \norm{\pint - \pint}_\infty & = 0 \, ,
\end{align}
therefore condition (iii) is fulfilled with
\begin{equation}
    T = Q  = 0, \qquad R = \id_\hilbert,\qquad  C = \pint \norm{\pint}^{-1}_\infty,
\end{equation}
and the operator $\mathcal{F}(z)$ from \cref{thm:swiderski} reduces to
\begin{equation}
	\mathcal{F}(z) = \Real \left( \begin{pmatrix}
		0 & -C \\ C & 0
	\end{pmatrix}
	\begin{pmatrix}
	 0 & \id_\hilbert \\ -\id_\hilbert &0
	\end{pmatrix}
	\right) =\Real \begin{pmatrix}
		C & 0 \\ 0 & C
\end{pmatrix} = \begin{pmatrix}
		C & 0 \\ 0 & C
\end{pmatrix} \, .
\end{equation}
As $\pint$ is strictly positive, so is $\mathcal{F}(z)$, and thus condition (iv) is also satisfied. Therefore, \cref{thm:swiderski} applies to $\tilde{J}^{(m)}$, which proves the claim.
\end{proof}

We can finally state our first main result concerning the class of operators from \cref{def:op}:
\begin{theorem}
    \label{thm:main-result}
	Let $H$ be the operator from \cref{def:op}; assume that $\Sigma$ is normal with bounded inverse, and $k\geq3$. Then $H$ is not self-adjoint, and its deficiency spaces $\dspace(\pm\iu)=\Ker(H^\ast\pm\iu)$ are given by
	\begin{equation}\label{eq:deficiency}
	    \dspace(\pm \iu) = \bigoplus_{m = 0}^{k-1} \left\{(P^{(m)}_n(\mp \iu) x)_{n \in \nnum} \, : \, x \in \hilbert \right\}  \cong \bigoplus_{m = 0}^{k-1} \hilbert  \, ,
	\end{equation}
    with $P^{(m)}_n(z)$ being the orthogonal polynomials of first kind of $J^{(m)}$, cf. \cref{lem:decomposition,def:orthogonal-polynomials}.    
    Moreover, all self-adjoint extensions $H_U$ of $H$ are parametrised by unitary operators $V: \dspace(+\iu) \to \dspace(-\iu)$.
\end{theorem}

\begin{proof}
    By \cref{lem:decomposition}, we can write
    \begin{equation}\label{eq:decomposition2}
        	H \simeq \bigoplus_{m = 0}^{k-1} J^{(m)} \, ,
    \end{equation}
    each $J^{(m)}$ being a block Jacobi operator with coefficients given by \cref{eq:coefficients1,eq:coefficients2}; by \cref{lem:trafo}, $J^{(m)} = \mathcal{U}^\ast \tilde{J}^{(m)} \mathcal{U}$, with $\tilde{J}^{(m)}$ being a block Jacobi operator with coefficients given by \cref{eq:coefficients1tilde,eq:coefficients2tilde} and $\mathcal{U}$ being the unitary operator from \cref{lem:trafo}. By \cref{prop:deficiency}, $\tilde{J}^{(m)}$ are maximally indeterminate; as $\mathcal{U}$ is unitary and leaves $\ell_0(\nnum, \hilbert)$ invariant, $J^{(m)}$ is also maximally indeterminate, and thus by \cref{lem:deficiency-spaces} its deficiency spaces are  given by
    \begin{equation}
        \dspace(\pm \iu ) = \Ker\left(\left(J^{(m)}\right)^\ast \pm \iu \right) =  \left\{(P^{(m)}_n(\mp \iu) x)_{n \in \nnum} \, : \, x \in \hilbert \right\} \cong \hilbert\, ;
    \end{equation}
together with \cref{eq:decomposition2}, this readily implies \cref{eq:deficiency}. The final claim about parametrisation of self-adjoint extensions follows from the standard theory of self-adjoint extensions, see e.g.~\cite[Theorem 13.10]{schmudgen-unboundedselfadjointoperators-2012}.
\end{proof}

This proves that all operators implementing $k$-photon matter--boson interactions as in \cref{def:op} admit multiple self-adjoint extensions for $k\geq3$, provided that the operator $\Sigma$ is normal and has a bounded inverse. 

In the case of finite-dimensional matter systems, $\dim\hilbert=d<\infty$, one can also recover some information about the spectrum. Hereafter, we use the following nomenclature: the discrete spectrum of an operator is the set of its isolated eigenvalues with finite multiplicity.

\begin{corollary}
    \label{cor:main-result-finite-dim}
	Under the same assumptions and notation of \cref{thm:main-result}, assume $\dim\hilbert = d<\infty$. 
    Then the spectrum of all self-adjoint extensions $H_U$ of $H$ is purely discrete.
\end{corollary}
\begin{proof}
   If $\dim\hilbert=d<\infty$, the spectrum of every self-adjoint extension of a maximally indeterminate block Jacobi operator on $\ell^2(\nnum,\hilbert)$ is purely discrete~\cite{budyka-selfadjointnessdiscretenessspectrum-2020}. Therefore, all self-adjoint extensions of each operator $J^{(m)}$ have discrete spectrum. 
   
    For $0 \leq m \leq k-1$, let $J_1^{(m)}$ be any self-adjoint extension of $J^{(m)}$. Then $H_1 = \bigoplus_{m = 0}^{k-1} J_1^{(m)}$ is a self-adjoint extension of $H$ with purely discrete spectrum.
    But by~\cite[Corollary 8.13]{schmudgen-unboundedselfadjointoperators-2012}, if an operator with finite deficiency indices admits one self-adjoint extension with discrete spectrum, \textit{all} its self-adjoint extensions have discrete spectrum. Therefore, the spectra of all self-adjoint extensions $H_U$ of $H$ are discrete.
\end{proof}

\subsection{Case 2: $\Sigma$ without bounded inverse}

The results obtained in the previous section (\cref{prop:ess-sa} for $k\leq2$ and \cref{thm:main-result,cor:main-result-finite-dim} for $k\geq3$) require the operator $\Sigma\in\bounded(\hilbert)$ in the expression of $H$ mediating the light--matter interaction (\cref{def:op}) to be normal and have bounded inverse. While the former assumption cannot be waived (cf. \cref{ex:kjc}), we now show how to extend these results to the case where the latter assumption fails (cf. \cref{ex:kdicke}).

\begin{proposition}
    \label{prop:main-result-inverse}
    Let $H$ be the operator from \cref{def:op}, assume that $\aint$ is normal, and $\Sigma\neq0$.
    Moreover, set  $\hilbert_0 = \Ker \Sigma$, and $\hilbert_1 = \hilbert_0^\perp$, and assume that the restriction $\Sigma_1 = \Sigma_{| \hilbert_1}$ of $\Sigma$ to $\hilbert_1$ has a bounded inverse. 
	Then the following statements hold:
    \begin{itemize}
        \item[(i)] If $k\leq2$, $H$ is self-adjoint;
        \item[(ii)] If $k\geq3$, then $H$ is not self-adjoint and its deficiency spaces are given by
    \begin{equation}
        \dspace(\pm \iu) \cong \bigoplus_{m = 0}^{k-1}\hilbert_1 \, .
    \end{equation}
    In particular, if $\dim\hilbert_1=d'<\infty$, then $H$ has deficiency indices $(kd',kd')$ and its self-adjoint extensions are parametrised by $kd'\times kd'$ unitary matrices. If, in addition, $\dim\hilbert<\infty$, then every such extension has purely discrete spectrum. In finite dimensions, these deficiency indices also follow from an elementary direct argument, presented in \cref{rem:finite-dim-argument}.
    \end{itemize}
\end{proposition}

\begin{proof}
It will be sufficient to prove both claims for the operator
\begin{equation}
        \tilde{H} = \overline{\tilde{H}_{\mathrm{min}}},\qquad \tilde{H}_{\mathrm{min}} = \Sigma \otimes (a^\ast)^k + \Sigma^\ast \otimes a^k + \id_\hilbert \otimes \omega a^\ast a;
    \end{equation}
    indeed, $H = \tilde{H} + \hmat\otimes\id_{L^2(\rnum)}$ differs from $\tilde{H}$ by a bounded self-adjoint operator. Self-adjointness (for $k\leq2$) and the deficiency indices (for $k\geq3$) are stable under bounded self-adjoint perturbations~\cite{schmudgen-unboundedselfadjointoperators-2012}, and are therefore inherited from $\tilde{H}$; the discreteness of the spectrum will be transferred to $H$ separately, at the end of the proof.

    Clearly $\hilbert_0=\Ker\Sigma$ is an invariant subspace for $\Sigma$; since $\Sigma$ is normal, it is actually a reducing subspace for $\Sigma$, that is, $\Sigma=0_{\hilbert_0}\oplus\Sigma_1$ with respect to the natural decomposition $\hilbert=\hilbert_0\oplus\hilbert_1$. Consequently,  $\hilbert\otimes L^2(\rnum)=\left(\hilbert_0\otimes L^2(\rnum)\right)\oplus\left(\hilbert_1\otimes L^2(\rnum)\right)$ and, with respect to this decomposition,
    \begin{equation}
        \label{proofeq:H-decomposition}
        \tilde{H} = \tilde{H}_0 \oplus \tilde{H}_1 \, ,
    \end{equation}
    where
    \begin{align}
        \tilde{H}_0 & =  \id_{\hilbert_0}\otimes \omega a^\ast a, \\
        \tilde{H}_1 & =  \id_{\hilbert_1}\otimes \omega a^\ast a+\Sigma_1 \otimes (a^\ast)^k + \Sigma_1^\ast\otimes a^k   ,
    \end{align}
    with $\Sigma_1$ having a bounded inverse by assumption. $\tilde{H}_0$ is clearly self-adjoint. 
    
    If $k\leq2$, $\tilde{H}_1$ is also self-adjoint by \cref{prop:ess-sa}, and therefore $\tilde{H}$ is self-adjoint. If $k\geq3$, notice that $\Sigma_1$ is normal because $\Sigma$ is; therefore, we can apply \cref{thm:main-result}: $\tilde{H}_1$ is not self-adjoint and its deficiency subspaces, and thus the deficiency subspaces of $\tilde{H}$, are isomorphic to $\bigoplus_{m = 0}^{k-1}\hilbert_1$. As anticipated, the same holds for $H$.
    
         Finally, we address the spectrum of the self-adjoint extensions when $\dim\hilbert<\infty$. 
        As $\tilde{H}_0$ is self-adjoint, using \cref{proofeq:H-decomposition}, any self-adjoint extension $\tilde{H}_U$ of $\tilde{H}$ decomposes as
        \begin{equation}
            \tilde{H}_U = \tilde{H}_0 \oplus \tilde{H}_{1,U} \, , 
        \end{equation}
        where $\tilde{H}_{1,U}$ is a self-adjoint extension of $\tilde{H}_1$.
        By \cref{cor:main-result-finite-dim}, $\tilde{H}_{1,U}$ has purely discrete spectrum; since $\dim\hilbert_0<\infty$, the operator $\tilde{H}_0=\id_{\hilbert_0}\otimes\omega a^\ast a$ has eigenvalues $\omega n$ of finite multiplicity, hence purely discrete spectrum as well.
        Therefore, the spectrum of $\tilde{H}_U$ is purely discrete. Since $\hmat\otimes\id_{L^2(\rnum)}$ is bounded and self-adjoint, it leaves the operator domain unchanged and preserves compactness of the resolvent; thus $\tilde{H}_U + \hmat\otimes\id_{L^2(\rnum)}$ is a self-adjoint extension of $H$ with purely discrete spectrum. As $H$ has finite deficiency indices, \cite[Corollary 8.13]{schmudgen-unboundedselfadjointoperators-2012} implies that every self-adjoint extension of $H$ has purely discrete spectrum.
\end{proof}

We finally remark that, if $\dim\hilbert=d<\infty$, then the restriction of $\Sigma$ to $\hilbert_1=\hilbert_0^\perp$ always has a bounded inverse. Therefore, \cref{prop:main-result-inverse} applies to \textit{all} operators from \cref{def:op}: if $\Sigma$ is normal and nonzero, $H$ is self-adjoint if and only if $k\leq2$.

\begin{remark}[A direct argument in finite dimensions]
    \label{rem:finite-dim-argument}
    When the matter space $\hilbert$ is finite-dimensional, the deficiency indices computed in \cref{prop:main-result-inverse} admit a more elementary derivation, which we sketch here. As $\Sigma$ is normal, there is an orthonormal basis $\{e_j\}_{j=1}^{d}$ of $\hilbert$ jointly diagonalising $\Sigma$ and $\Sigma^\ast$, with $\Sigma e_j=\lambda_j e_j$. Leaving aside the matter term $\hmat\otimes\id_{L^2(\rnum)}$, the operator decomposes along this basis as
    \begin{equation}
    \label{eq:direct-sum-aint}
        \id_{\hilbert}\otimes\omega a^\ast a + \Sigma\otimes(a^\ast)^k + \Sigma^\ast\otimes a^k \;\simeq\; \bigoplus_{j=1}^{d}\left(\omega a^\ast a + \lambda_j (a^\ast)^k + \bar\lambda_j a^k\right),
    \end{equation}
    a direct sum of scalar operators on $L^2(\rnum)$ of the form studied in \cite{ashhab-finitedimensionalapproximationsgeneralized-2026,fischer-selfadjointrealizationshigherorder-2026,fischer-essentiallysingularlimits-2026}. Whenever $\lambda_j\neq0$, the deficiency indices are $(k,k)$ by \cite{fischer-selfadjointrealizationshigherorder-2026}; when $\lambda_j=0$, the summand is the free Hamiltonian $\omega a^\ast a$, which is essentially self-adjoint. As the number of nonzero eigenvalues, counted with multiplicity, equals $d'=\dim(\Ker\Sigma)^\perp$, the direct sum has deficiency indices $(kd',kd')$; adding back the bounded self-adjoint operator $\hmat\otimes\id_{L^2(\rnum)}$ leaves them unchanged.

   This argument extends, in principle, to an arbitrary bounded normal $\Sigma$ by replacing the eigendecomposition with the spectral theorem and the direct sum with a direct integral over the spectrum $\sigma(\Sigma)\subseteq\cnum$ of $\Sigma$; a careful analysis of the resulting direct integral would be needed to make this rigorous. Moreover, this argument depends intrinsically on the particular structure of the model and the resulting decomposition into squeezing operators. The block Jacobi decomposition, by contrast, is in principle applicable to broader polynomial interactions.
  
    Finally, we note that a similar argument can be used to obtain a lower bound on the deficiency indices even if $\aint$ is non-normal: if $\aint$ and $\aint^\ast$ share $d'$ linearly independent eigenvectors with non-zero eigenvalues, then the deficiency indices are at least $n_\pm \geq k d'$.
\end{remark}

\section{Examples}
\label{sec:examples}

We conclude by applying the results of \cref{sec:main-result} to some examples of physical interest: the $k$-photon quantum Rabi model (\cref{ex:krabi}) which served as the original motivation for this work, and its generalisation to multiple spins (\cref{ex:kdicke}). Both operators turn out to be self-adjoint for $k\leq2$ and not self-adjoint, with multiple self-adjoint extensions, for $k\geq3$. We finally discuss an example illustrating that the assumption of normality of the matrix $\Sigma$ mediating the light--matter interaction \textit{cannot} be waived in general: the $k$-photon Jaynes--Cummings model (\cref{ex:kjc}), which turns out to be self-adjoint for all $k$.

\begin{example}[$k$-photon quantum Rabi model]
\label{ex:krabi}
We consider again the $k$-photon Rabi model discussed in \cref{sec:intro}. On the Hilbert space $\cnum^2\otimes L^2(\rnum)$, we set
\begin{equation}  \label{eq:krabi}
    H =  \frac{\Omega}{2}\sigma_z \otimes  \id_{L^2(\rnum)} + \id_{\cnum^2} \otimes \omega a^\ast a +  g\sigma_x\otimes \left(a^k + (a^\ast)^k \right) ,
\end{equation}
where $\Omega$ is the energy splitting of the spin, and $g > 0$ is the strength of the light--matter interaction. The operator is initially defined on the minimal domain $\cnum^2\otimes\domain_0$ and then extended by closure. This is a specific instance of the class in \cref{def:op} obtained by setting
\begin{equation}
    \aint = g \sigma_x \, , \qquad \hmat = \frac{\Omega}{2} \sigma_z \, .
\end{equation}
Clearly $\Sigma$ is normal and invertible and the conditions of \cref{prop:ess-sa} and \cref{thm:main-result} are fulfilled: therefore, for $k\leq2$ it is self-adjoint, and for $k\geq3$ it has deficiency indices $(2k,2k)$ and thus admits multiple self-adjoint extensions parametrised by $2k\times2k$ unitary matrices. This result is compatible with the one found by Braak for $k=3$ via explicit computation in the Bargmann space \cite{braak-$k$photonquantumrabi-2025}.
Moreover, by \cref{cor:main-result-finite-dim}, the spectra of all its self-adjoint extensions are purely discrete.
\end{example}

\begin{example}[$k$-photon Dicke model]\label{ex:kdicke}
We consider the $N$-spin generalisation of the $k$-photon quantum Rabi model, which---in analogy with the usual nomenclature in the physics literature for the single-photon case---we will denote as the \textit{$k$-photon Dicke model}. Given $N\in\nnum$, on the Hilbert space $\cnum^{2^N}\otimes L^2(\rnum)$ we set
\begin{equation}
        H =    \frac{\Omega}{2}\Sigma_z\otimes  \id_{L^2(\rnum)} + \id_{\cnum^{2^N}}\otimes\omega a^\ast a + \Sigma_x\otimes g \left(a^k + (a^\ast)^k \right)   ,
    \end{equation}
    where $\Sigma_x,\Sigma_z$ are collective spin operators on $(\cnum^2)^{\otimes N} \cong \cnum^{2^N}$:
    \begin{equation}
        \Sigma_z  = \sum_{i = 1}^N \sigma_z^{(i)} \, , \quad \Sigma_x  = \sum_{i = 1}^N \sigma_x^{(i)} \,  ,
    \end{equation}
    where 
    \begin{equation}
       \sigma_z^{(i)}=\id_{\cnum^2}\otimes\cdots\otimes\overbrace{\sigma_z}^{i\text{th}}\otimes\cdots\otimes\id_{\cnum^2},\qquad \sigma_x^{(i)}=\id_{\cnum^2}\otimes\cdots\otimes\overbrace{\sigma_x}^{i\text{th}}\otimes\cdots\otimes\id_{\cnum^2}.
    \end{equation}
   Physically, this operator models a family of $N$ spins separately interacting with a single-mode boson field via $k$-photon interactions; the parameters $g,\Omega$ have analogous interpretations as in \cref{ex:krabi}. This corresponds to the model in \cref{def:op} with 
   \begin{equation}
       \Sigma=g\Sigma_x,\qquad \hmat=\frac{\Omega}{2}\Sigma_z.
   \end{equation}
    By direct scrutiny, one can see that $\Sigma_x$ is normal, but it is invertible if and only if $N$ is odd, thus making \cref{prop:ess-sa} and \cref{thm:main-result} not directly applicable. In fact, one can show the following:
    \begin{equation}
        \dim \Ker \Sigma_x = \begin{cases}
            0 & N \text{ odd} \\
             \dbinom{N}{N/2}
            & N \text{ even}
        \end{cases} \, .
    \end{equation}
Heuristically, if $N$ is even, there are $\dbinom{N}{N/2}$ possibilities to create a state with total spin in the $x$ direction equal to zero. By applying \cref{prop:main-result-inverse}, it follows that $H$ is self-adjoint for $k\leq2$, and for $k\geq3$ its deficiency indices are equal to
    \begin{equation}
        n_\pm = \begin{cases}
            2^N k & N \text{ odd}; \\
            \left(2^N -  
            \dbinom{N}{N/2}
            \right) k & N \text{ even}.
        \end{cases} \, 
    \end{equation}
    These indices can equivalently be read off by diagonalising $\Sigma_x$, cf. \cref{rem:finite-dim-argument}.
    Finally, we remark that the above results stay valid if one adds additional terms to the energy of the spin subsystems, e.g. incorporating spin--spin interactions.
  \end{example}

\begin{example}[$k$-photon Jaynes--Cummings model]\label{ex:kjc}
As anticipated, the assumption that $\Sigma$ be normal cannot be relaxed. To illustrate this, we consider the following model on $\cnum^2\otimes L^2(\rnum)$:
    \begin{equation}\label{eq:kjc}
        H =  \id_{\cnum^2} \otimes \omega a^\ast a +  \frac{\Omega}{2}\sigma_z \otimes \id_{L^2(\rnum)} + g\left(\sigma_- \otimes (a^\ast)^k + \sigma_+ \otimes a^k\right) , 
    \end{equation}
    where $\sigma_\pm$ are spin ladder operators:
    \begin{equation}
        \sigma_{\pm} = \frac{1}{2} \left( \sigma_x \pm \iu \sigma_y\right) \, .
    \end{equation}
For $k=1$, this reduces to the so-called Jaynes--Cummings model~\cite{larson-jaynescummingsmodel-2024}, which corresponds to the \textit{rotating-wave approximation} (RWA) of the quantum Rabi model~\cite{burgarth2024taming,richter-quantifyingrotatingwaveapproximation-2026}, obtained by neglecting the fast-oscillating terms $\sigma_+\otimes a^\ast$ and $\sigma_-\otimes a$ (compare \cref{eq:kjc} with \cref{eq:krabi}). We thus denote this as the \textit{$k$-photon Jaynes--Cummings model}~\cite{sukumar-multiphonongeneralisationjaynescummings-1981,mir-amplitudesquaredsqueezingmultiphoton-1993,el-orany-evolutionsuperpositionsqueezed-2003}.

Again, this model belongs to the class from \cref{def:op} with
\begin{equation}
    \Sigma=g\sigma_-,\qquad \hmat=\frac{\Omega}{2}\sigma_z;
\end{equation}
however, $\sigma_-$ is \textit{not} normal, as $\sigma_-^\ast = \sigma_+$ and the matrices $\sigma_-,\sigma_+$ do not commute. This makes all results of \cref{sec:main-result} inapplicable. In fact, we claim that $H$ is self-adjoint for all choices of $k$.

To this end, we follow~\cite{hillery-photonnumberdivergence-1990} and introduce the operator $\mathcal{N}_k = \id_{\cnum^2}\otimes a^\ast a + k \sigma_+ \sigma_-\otimes\id_{L^2(\rnum)}$, which represents a modified excitation number of the system. The Hilbert space of the system decomposes as follows:
    \begin{equation}
        \cnum^2 \otimes L^2(\rnum)\cong \bigoplus_{m = 0}^\infty \hilbert_m \, , 
    \end{equation}
    where $\hilbert_m$ is the $m$th eigenspace of $\mathcal{N}_k$:
    \begin{equation}
        \hilbert_m = \begin{cases}
            \mspan( e_0\otimes \phi_m) & 0 \leq m \leq k-1 \\
            \mspan( e_0\otimes \phi_m,  e_1 \otimes \phi_{m-k}) & m \geq k .
        \end{cases}
    \end{equation}
We claim that $\mathcal{N}_k $ is a conserved quantity for $H$, that is, each $\hilbert_m$ is an invariant subspace of $H$.
    We start with $m \leq k-1$, and write $\psi \in \hilbert_m$ as $e_0\otimes c_0 \phi_m$ for some $c_0 \in \cnum$.
    It immediately follows $H \psi = \big(\omega m-\frac{\Omega}{2}\big) \psi$, and hence
    \begin{equation}
        H \hilbert_m  \subset \hilbert_m \quad \forall m \leq k-1 \, . 
    \end{equation}
    Now we turn to $m \geq k$, and write $\psi \in \hilbert_m$ as $  e_0 \otimes c_0 \phi_m+  e_1\otimes c_1 \phi_{m-k}$.
    The diagonal part of $H$ (the terms proportional to $\omega$ and $\Omega/2$) clearly preserves $\hilbert_m$. As for the interaction, using the definition of $\beta_n$ from \cref{lem:decomposition}, we obtain
    \begin{equation}
        g\left(\sigma_- \otimes (a^\ast)^k + \sigma_+ \otimes a^k\right) \psi =  e_1 \otimes c_0 g\beta_{m-k} \phi_{m-k} +  e_0 \otimes c_1 g\beta_{m-k} \phi_m \in \hilbert_m \, , 
    \end{equation}
    and therefore $H$ leaves all subspaces $\hilbert_m$ invariant. This readily implies
    \begin{equation}
        H = \bigoplus_{m = 0}^\infty H_m,
    \end{equation}
    with $H_m$ being the restriction of $H$ to $\hilbert_m$, which is clearly bounded as each $\hilbert_m$ is finite-dimensional. As such, $H$ is the direct sum of self-adjoint operators and is thus self-adjoint~\cite[Theorem 2.23]{teschl-mathematicalmethodsquantum-2009}.
\end{example}

\subsection*{Acknowledgements} The authors gratefully acknowledge discussions with Daniel Braak.

\printbibliography

@article{burgarth2024taming,
  title = {Taming the {{Rotating Wave Approximation}}},
  author = {Burgarth, Daniel and Facchi, Paolo and Hillier, Robin and Ligabò, Marilena},
  date = {2024-02-21},
  journaltitle = {Quantum},
  volume = {8},
  pages = {1262},
  doi = {10.22331/q-2024-02-21-1262},
  url = {https://quantum-journal.org/papers/q-2024-02-21-1262/},
}

@book{akhiezer-classicalmomentproblem-2020,
  title = {The {{Classical Moment Problem}} and {{Some Related Questions}} in {{Analysis}}},
  author = {Akhiezer, N. I.},
  date = {2020-01},
  publisher = {{Society for Industrial and Applied Mathematics}},
  location = {Philadelphia, PA},
  doi = {10.1137/1.9781611976397},
  isbn = {978-1-61197-638-0 978-1-61197-639-7}
}

@article{ashhab-finitedimensionalapproximationsgeneralized-2026,
  title = {Finite-Dimensional Approximations of Generalized Squeezing},
  author = {Ashhab, Sahel and Fischer, Felix and Lonigro, Davide and Braak, Daniel and Burgarth, Daniel},
  date = {2026-01-02},
  journaltitle = {Physical Review A},
  shortjournal = {Phys. Rev. A},
  volume = {113},
  number = {1},
  pages = {013703},
  issn = {2469-9926, 2469-9934},
  doi = {10.1103/9vwp-f35c}
}

@article{bazavan-squeezingtrisqueezingquadsqueezing-2026,
  title = {Squeezing, Trisqueezing and Quadsqueezing in a Hybrid Oscillator–Spin System},
  author = {Băzăvan, O. and Saner, S. and Webb, D. J. and Ainley, E. M. and Drmota, P. and Nadlinger, D. P. and Araneda, G. and Lucas, D. M. and Ballance, C. J. and Srinivas, R.},
  sortname = {Bazavan, O. and Saner, S. and Webb, D. J. and Ainley, E. M. and Drmota, P. and Nadlinger, D. P. and Araneda, G. and Lucas, D. M. and Ballance, C. J. and Srinivas, R.},
  date = {2026-05-01},
  journaltitle = {Nature Physics},
  shortjournal = {Nat. Phys.},
  volume={22},
  pages = {757--762},
  issn = {1745-2473, 1745-2481},
  doi = {10.1038/s41567-026-03222-6}
}

@article{bencheikh-demonstratingquantumproperties-2022,
  title = {Demonstrating Quantum Properties of Triple Photons Generated by $\chi^3$-processes},
  author = {Bencheikh, Kamel and Cenni, Marina F. B. and Oudot, Enky and Boutou, Véronique and Félix, Corinne and Prades, Joel Compte and Vernay, Augustin and Bertrand, Julien and Bassignot, Florent and Chauvet, Mathieu and Bussières, Félix and Zbinden, Hugo and Levenson, Ariel and Boulanger, Benoît},
  date = {2022-10-10},
  journaltitle = {The European Physical Journal D},
  shortjournal = {Eur. Phys. J. D},
  volume = {76},
  number = {10},
  pages = {186},
  issn = {1434-6079},
  doi = {10.1140/epjd/s10053-022-00514-3}
}

@article{boutin-effecthigherordernonlinearities-2017,
  title = {Effect of {{Higher-Order Nonlinearities}} on {{Amplification}} and {{Squeezing}} in {{Josephson Parametric Amplifiers}}},
  author = {Boutin, Samuel and Toyli, David M. and Venkatramani, Aditya V. and Eddins, Andrew W. and Siddiqi, Irfan and Blais, Alexandre},
  date = {2017-11-15},
  journaltitle = {Physical Review Applied},
  shortjournal = {Phys. Rev. Applied},
  volume = {8},
  number = {5},
  pages = {054030},
  issn = {2331-7019},
  doi = {10.1103/PhysRevApplied.8.054030}
}

@InBook{braak-$k$photonquantumrabi-2025,
author="Braak, Daniel",
editor="Takagi, Tsuyoshi
and Wakayama, Masato
and Kunihiro, Noboru
and Tanaka, Keisuke
and Kimoto, Kazufumi
and Kudo, Momonari",
title="The $k$-Photon Quantum Rabi Model",
bookTitle="Mathematical Foundations for Post-Quantum Cryptography: Crypto-Math CREST",
year="2026",
publisher="Springer Nature Singapore",
address="Singapore",
pages="75--87",
isbn="978-981-96-1218-5",
doi="10.1007/978-981-96-1218-5_5",
url="https://doi.org/10.1007/978-981-96-1218-5_5"
}

@article{braak-integrabilityrabimodel-2011,
  title = {Integrability of the {{Rabi Model}}},
  author = {Braak, D.},
  date = {2011-08-29},
  journaltitle = {Physical Review Letters},
  shortjournal = {Phys. Rev. Lett.},
  volume = {107},
  number = {10},
  pages = {100401},
  issn = {0031-9007, 1079-7114},
  doi = {10.1103/PhysRevLett.107.100401}
}

@article{braeutigam-deficiencyindicesoperators-2019,
  title = {Deficiency Indices of the Operators Generated by Infinite {{Jacobi}} Matrices with Operator Entries},
  author = {Braeutigam, I. N. and Mirzoev, K. A.},
  date = {2019-06-04},
  journaltitle = {St. Petersburg Mathematical Journal},
  shortjournal = {St. Petersburg Math. J.},
  volume = {30},
  number = {4},
  pages = {621--638},
  issn = {1061-0022, 1547-7371},
  doi = {10.1090/spmj/1562}
}

@article{budyka-deficiencyindicesblock-2024,
  title = {Deficiency {{Indices}} of {{Block Jacobi Matrices}}: {{Survey}}},
  shorttitle = {Deficiency {{Indices}} of {{Block Jacobi Matrices}}},
  author = {Budyka, V. S. and Malamud, M. M. and Mirzoev, K. A.},
  date = {2024-01},
  journaltitle = {Journal of Mathematical Sciences},
  shortjournal = {J Math Sci},
  volume = {278},
  number = {1},
  pages = {39--54},
  issn = {1072-3374, 1573-8795},
  doi = {10.1007/s10958-024-06904-9}
}

@article{budyka-deficiencyindicesdiscreteness-2022a,
  title = {Deficiency Indices and Discreteness Property of Block {{Jacobi}} Matrices and {{Dirac}} Operators with Point Interactions},
  author = {Budyka, Viktoriya S. and Malamud, Mark M.},
  date = {2022-02-01},
  journaltitle = {Journal of Mathematical Analysis and Applications},
  shortjournal = {Journal of Mathematical Analysis and Applications},
  volume = {506},
  number = {1},
  pages = {125582},
  issn = {0022-247X},
  doi = {10.1016/j.jmaa.2021.125582}
}

@article{budyka-selfadjointnessdiscretenessspectrum-2020,
  title = {Self-{{Adjointness}} and {{Discreteness}} of the {{Spectrum}}     of {{Block Jacobi Matrices}}},
  author = {Budyka, V. S. and Malamud, M. M.},
  date = {2020-10-01},
  journaltitle = {Mathematical Notes},
  shortjournal = {Math Notes},
  volume = {108},
  number = {3},
  pages = {445--450},
  issn = {1573-8876},
  doi = {10.1134/S000143462009014X}
}

@article{chang-observationthreephotonspontaneous-2020,
  title = {Observation of {{Three-Photon Spontaneous Parametric Down-Conversion}} in a {{Superconducting Parametric Cavity}}},
  author = {Chang, C. W. Sandbo and Sabín, Carlos and Forn-Díaz, P. and Quijandría, Fernando and Vadiraj, A. M. and Nsanzineza, I. and Johansson, G. and Wilson, C. M.},
  date = {2020-01-16},
  journaltitle = {Physical Review X},
  shortjournal = {Phys. Rev. X},
  volume = {10},
  number = {1},
  pages={011011},
  publisher = {American Physical Society (APS)},
  issn = {2160-3308},
  doi = {10.1103/physrevx.10.011011}
}

@article{corona-thirdorderspontaneousparametric-2011,
  title = {Third-Order Spontaneous Parametric down-Conversion in Thin Optical Fibers as a Photon-Triplet Source},
  author = {Corona, María and Garay-Palmett, Karina and U’Ren, Alfred B.},
  date = {2011-09-15},
  journaltitle = {Physical Review A},
  shortjournal = {Phys. Rev. A},
  volume = {84},
  number = {3},
  pages={033823},
  publisher = {American Physical Society (APS)},
  issn = {1050-2947, 1094-1622},
  doi = {10.1103/physreva.84.033823}
}

@article{eriksson-universalcontrolbosonic-2024,
  title = {Universal Control of a Bosonic Mode via Drive-Activated Native Cubic Interactions},
  author = {Eriksson, Axel M. and Sépulcre, Théo and Kervinen, Mikael and Hillmann, Timo and Kudra, Marina and Dupouy, Simon and Lu, Yong and Khanahmadi, Maryam and Yang, Jiaying and Castillo-Moreno, Claudia and Delsing, Per and Gasparinetti, Simone},
  date = {2024-03-21},
  journaltitle = {Nature Communications},
  shortjournal = {Nat Commun},
  volume = {15},
  number = {1},
  pages = {2512},
  issn = {2041-1723},
  doi = {10.1038/s41467-024-46507-1}
}

@online{fischer-essentiallysingularlimits-2026,
  title = {Essentially Singular Limits of {{Jacobi}} Operators and Applications to Higher-Order Squeezing},
  author = {Fischer, Felix and Burgarth, Daniel and Lonigro, Davide},
  date = {2026-05-20},
  eprint = {2605.21355},
  eprinttype = {arXiv},
  eprintclass = {math-ph},
  doi = {10.48550/arXiv.2605.21355},
  pubstate = {prepublished}
}

@article{fischer-selfadjointrealizationshigherorder-2026,
  title = {Self-Adjoint Realizations of Higher-Order Squeezing Operators},
  author = {Fischer, Felix and Burgarth, Daniel and Lonigro, Davide},
  date = {2026},
  volume={59},
  pages={255203},
  journaltitle = {Journal of Physics A: Mathematical and Theoretical},
  shortjournal = {J. Phys. A: Math. Theor.},
  issn = {1751-8121},
  doi = {10.1088/1751-8121/ae7acb}
}

@online{gregory-foursixphotonstimulated-2026,
  title = {Four- and Six-Photon Stimulated {{Raman}} Transitions for Coherent Qubit and Qudit Operations},
  author = {Gregory, Gabriel J. and Ritchie, Evan R. and Quinn, Alex and Brudney, Sean and Wineland, David J. and Allcock, David T. C. and O'Reilly, Jameson},
  date = {2026-02-20},
  eprint = {2602.18567},
  eprinttype = {arXiv},
  eprintclass = {quant-ph},
  doi = {10.48550/arXiv.2602.18567},
  pubstate = {prepublished}
}

@article{menard-emissionphotonmultiplets-2022,
  title = {Emission of {{Photon Multiplets}} by a dc-{{Biased Superconducting Circuit}}},
  author = {Ménard, G. C. and Peugeot, A. and Padurariu, C. and Rolland, C. and Kubala, B. and Mukharsky, Y. and Iftikhar, Z. and Altimiras, C. and Roche, P. and Le Sueur, H. and Joyez, P. and Vion, D. and Esteve, D. and Ankerhold, J. and Portier, F.},
  sortname = {Menard, G. C. and Peugeot, A. and Padurariu, C. and Rolland, C. and Kubala, B. and Mukharsky, Y. and Iftikhar, Z. and Altimiras, C. and Roche, P. and Le Sueur, H. and Joyez, P. and Vion, D. and Esteve, D. and Ankerhold, J. and Portier, F.},
  date = {2022-04-08},
  journaltitle = {Physical Review X},
  shortjournal = {Phys. Rev. X},
  volume = {12},
  number = {2},
  pages={021006},
  publisher = {American Physical Society (APS)},
  issn = {2160-3308},
  doi = {10.1103/physrevx.12.021006}
}

@article{moszynski-barriernonsubordinacyabsolutely-2025,
  title = {Barrier {{Nonsubordinacy}} and {{Absolutely Continuous Spectrum}} of {{Block Jacobi Matrices}}},
  author = {Moszyński, Marcin and Świderski, Grzegorz},
sortname = {Moszynski, Marcin and Swiderski, Grzegorz},
  date = {2025-11-17},
  journaltitle = {Constructive Approximation},
  shortjournal = {Constr Approx},
  volume={63},
  pages={475–504},
  issn = {1432-0940},
  doi = {10.1007/s00365-025-09724-5}
}

@book{rudin-functionalanalysis-2007,
  title = {Functional Analysis},
  author = {Rudin, Walter},
  date = {2007},
  series = {International Series in Pure and Applied Mathematics},
  edition = {2. ed.},
  publisher = {McGraw-Hill},
  location = {Boston, Mass.},
  isbn = {978-0-07-054236-5},
  pagetotal = {424}
}

@book{schmudgen-momentproblem-2017,
  title = {The {{Moment Problem}}},
  author = {Schmüdgen, Konrad},
  date = {2017},
  series = {Graduate {{Texts}} in {{Mathematics}}},
  volume = {277},
  publisher = {Springer International Publishing},
  location = {Cham},
  doi = {10.1007/978-3-319-64546-9},
  isbn = {978-3-319-64545-2 978-3-319-64546-9}
}

@book{schmudgen-unboundedselfadjointoperators-2012,
  title = {Unbounded {{Self-adjoint Operators}} on {{Hilbert Space}}},
  author = {Schmüdgen, Konrad},
  date = {2012},
  series = {Graduate {{Texts}} in {{Mathematics}}},
  volume = {265},
  publisher = {Springer Netherlands},
  location = {Dordrecht},
  doi = {10.1007/978-94-007-4753-1},
  isbn = {978-94-007-4752-4 978-94-007-4753-1}
}

@article{schulz-baldes-rotationnumbersjacobi-2007,
  title = {Rotation Numbers for {{Jacobi}} Matrices with Matrix Entries},
  author = {Schulz-Baldes, Hermann},
  journal={Mathematical Physics Electronic Journal},
  volume={13},
  pages={5},
  year={2007},
  url={https://eudml.org/doc/129089}
}

@article{swiderski-asymptoticzerosdistribution-2025,
  title = {Asymptotic Zeros' Distribution of Orthogonal Polynomials with Unbounded Recurrence Coefficients},
  author = {Świderski, Grzegorz and Trojan, Bartosz},
  sortname = {Swiderski, Grzegorz and Trojan, Bartosz},
  date = {2025-12-15},
  journaltitle = {Journal of Functional Analysis},
  shortjournal = {Journal of Functional Analysis},
  volume = {289},
  number = {12},
  pages = {111162},
  issn = {0022-1236},
  doi = {10.1016/j.jfa.2025.111162}
}

@article{swiderski-spectralpropertiesblock-2018,
  title = {Spectral {{Properties}} of {{Block Jacobi Matrices}}},
  author = {Świderski, Grzegorz},
  sortname = {Swiderski, Grzegorz},
  date = {2018-10},
  journaltitle = {Constructive Approximation},
  shortjournal = {Constr Approx},
  volume = {48},
  number = {2},
  pages = {301--335},
  issn = {0176-4276, 1432-0940},
  doi = {10.1007/s00365-018-9420-z}
}

@book{teschl-jacobioperatorscompletely-1999,
  title = {Jacobi {{Operators}} and {{Completely Integrable Nonlinear Lattices}}},
  author = {Teschl, Gerald},
  date = {1999-10-05},
  series = {Mathematical {{Surveys}} and {{Monographs}}},
  volume = {72},
  publisher = {American Mathematical Society},
  location = {Providence, Rhode Island},
  doi = {10.1090/surv/072},
  isbn = {978-0-8218-1940-1 978-1-4704-1299-9}
}

@book{teschl-mathematicalmethodsquantum-2009,
  title = {Mathematical {{Methods}} in {{Quantum Mechanics}}},
  author = {Teschl, Gerald},
  date = {2009},
  series = {Graduate {{Studies}} in {{Mathematics}}},
  volume = {99},
  publisher = {American Mathematical Society},
  location = {Providence, Rhode Island},
  doi = {10.1090/gsm/157}
}

@Article{ashhab-fractionalsqueezingspectra-2026,
doi = {10.1088/1751-8121/ae7681},
url = {https://doi.org/10.1088/1751-8121/ae7681},
year = {2026},
publisher = {IOP Publishing},
volume = {59},
number = {25},
pages = {255301},
author = {Ashhab, Sahel},
title = {Fractional squeezing: spectra and dynamics from generalized squeezing Hamiltonian with fractional orders},
journal = {Journal of Physics A: Mathematical and Theoretical}
}

@book{berezanskii-expansioneigenfunctionsselfadjoint-1968,
  title = {Expansion in {{Eigenfunctions}} of {{Self-adjoint Operators}}},
  author = {Berezanskii, Ju M},
  date = {1968},
  series = {Translations of {{Mathematical}}                         {{Monographs}}},
  volume = {17},
  publisher = {American Mathematical                     Society}
}

@article{xie-quantumrabimodel-2017,
  title = {The Quantum {{Rabi}} Model: Solution and Dynamics},
  shorttitle = {The Quantum {{Rabi}} Model},
  author = {Xie, Qiongtao and Zhong, Honghua and Batchelor, Murray T and Lee, Chaohong},
  date = {2017-03-17},
  journaltitle = {Journal of Physics A: Mathematical and Theoretical},
  shortjournal = {J. Phys. A: Math. Theor.},
  volume = {50},
  number = {11},
  pages = {113001},
  issn = {1751-8113, 1751-8121},
  doi = {10.1088/1751-8121/aa5a65}
}

@article{richter-quantifyingrotatingwaveapproximation-2026,
  title = {Quantifying the Rotating-Wave Approximation of the {{Dicke}} Model},
  author = {Richter, Leonhard and Burgarth, Daniel and Lonigro, Davide},
  date = {2026-02},
  journaltitle = {Journal of Physics A: Mathematical and Theoretical},
  shortjournal = {J. Phys. A: Math. Theor.},
  volume = {59},
  number = {7},
  pages = {075203},
  publisher = {IOP Publishing},
  issn = {1751-8121},
  doi = {10.1088/1751-8121/ae42a3}
}

@book{larson-jaynescummingsmodel-2024,
  title = {The Jaynes–Cummings Model and Its Descendants (Second Edition)},
  author = {Larson, Jonas and Mavrogordatos, Themistoklis},
  date = {2024},
  publisher = {IOP Publishing},
  doi = {10.1088/978-0-7503-6452-2},
  isbn = {978-0-7503-6452-2}
}

@article{hillery-photonnumberdivergence-1990,
  title = {Photon Number Divergence in the Quantum Theory of {\mkbibemph{n}}-Photon down Conversion},
  author = {Hillery, Mark},
  date = {1990-07-01},
  journaltitle = {Physical Review A},
  shortjournal = {Phys. Rev. A},
  volume = {42},
  number = {1},
  pages = {498--502},
  issn = {1050-2947, 1094-1622},
  doi = {10.1103/PhysRevA.42.498}
}

@book{olver-asymptoticsspecialfunctions-2010,
  title = {Asymptotics and Special Functions},
  author = {Olver, Frank W. J.},
  date = {2010},
  series = {{{AKP}} Classics},
  edition = {Reprinted},
  publisher = {CRC Press},
  location = {Boca Raton, Fla.},
  isbn = {978-1-56881-069-0},
  pagetotal = {572}
}

@article{el-orany-evolutionsuperpositionsqueezed-2003,
  title = {On the Evolution of Superposition of Squeezed Displaced Number States with the Multiphoton {{Jaynes}}–{{Cummings}} Model},
  author = {El-Orany, Faisal A. A. and Obada, A.-S.},
  date = {2003-01},
  journaltitle = {Journal of Optics B: Quantum and Semiclassical Optics},
  shortjournal = {J. Opt. B: Quantum Semiclass. Opt.},
  volume = {5},
  number = {1},
  pages = {60},
  issn = {1464-4266},
  doi = {10.1088/1464-4266/5/1/309}
}

@article{sukumar-multiphonongeneralisationjaynescummings-1981,
  title = {Multi-Phonon Generalisation of the {{Jaynes-Cummings}} Model},
  author = {Sukumar, C. V. and Buck, B.},
  date = {1981-06-01},
  journaltitle = {Physics Letters A},
  shortjournal = {Physics Letters A},
  volume = {83},
  number = {5},
  pages = {211--213},
  issn = {0375-9601},
  doi = {10.1016/0375-9601(81)90825-2}
}

@article{mir-amplitudesquaredsqueezingmultiphoton-1993,
  title = {Amplitude-Squared Squeezing in the Multiphoton {{Jaynes-Cummings}} Model: {{Effects}} of the Atomic Coherent States and Detuning},
  shorttitle = {Amplitude-Squared Squeezing in the Multiphoton {{Jaynes-Cummings}} Model},
  author = {Mir, Mubeen A.},
  date = {1993-05-01},
  journaltitle = {Physical Review A},
  shortjournal = {Phys. Rev. A},
  volume = {47},
  number = {5},
  pages = {4384--4391},
  issn = {1050-2947, 1094-1622},
  doi = {10.1103/PhysRevA.47.4384}
}

@article{zhang-2modekphotonquantum-2017,
  title = {On the 2-Mode and k-Photon Quantum {{Rabi}} Models},
  author = {Zhang, Yao-Zhong},
  date = {2017-05},
  journaltitle = {Reviews in Mathematical Physics},
  shortjournal = {Rev. Math. Phys.},
  volume = {29},
  number = {04},
  pages = {1750013},
  publisher = {World Scientific Pub Co Pte Lt},
  issn = {0129-055X, 1793-6659},
  doi = {10.1142/s0129055x17500131}
}

@article{Deng26,
author = {Xiaowei Deng and Yanyan Cai and Zhongchu Ni and Libo Zhang and Song Liu and Yuan Xu and Dapeng Yu},
journal = {Optica},
number = {7},
pages = {1334--1339},
publisher = {Optica Publishing Group},
title = {Generating and characterizing generalized squeezed states in a superconducting microwave cavity},
volume = {13},
year = {2026},
url = {https://opg.optica.org/optica/abstract.cfm?URI=optica-13-7-1334},
doi = {10.1364/OPTICA.596662},
}

\end{document}